  \documentclass[multi]{EngC} 

  \usepackage[rightcaption,raggedright]{sidecap}
  \usepackage{framed}         
  \usepackage{soul}           
 \usepackage{cite}  

  \usepackage{rotating}
  \usepackage{floatpag}
  \rotfloatpagestyle{empty}

  \usepackage{amsthm}
  \usepackage{graphicx}
  \usepackage{longtable}
    
  \usepackage{multind}
  \makeindex{subject}

  \theoremstyle{plain}
  \newtheorem{theorem}{Theorem}[chapter]
  \newtheorem{lemma}[theorem]{Lemma}
  \newtheorem{corollary}{Corollary}

  \theoremstyle{definition}
  \newtheorem{definition}[theorem]{Definition}

  \theoremstyle{remark}

  \usepackage{amsmath}
\usepackage{amssymb}
\usepackage{amsfonts}
\usepackage{amsthm}
\usepackage{mathrsfs}
\usepackage{xspace}
\usepackage{bm}
\usepackage{fancyref}
\usepackage{textcomp}
\usepackage[Symbol]{upgreek}
\usepackage{mleftright}
\usepackage{enumerate}
\usepackage{mathtools}

\newcommand{\congmat}[2]{[#1\ \ #2]}
\newcommand{\safemath}[2]{\newcommand{#1}{\ensuremath{#2}\xspace}}
\newcommand{\Nsafemath}[2]{\renewcommand{#1}{\ensuremath{#2}\xspace}}
\newcommand{\ind}[1]{\chi_{#1}}	

\safemath{\sca}{a}
\safemath{\scb}{b}
\safemath{\scc}{c}
\safemath{\scd}{d}
\safemath{\sce}{e}
\safemath{\scf}{f}
\safemath{\scg}{g}
\safemath{\sch}{h}
\safemath{\sci}{i}
\safemath{\scj}{j}
\safemath{\sck}{k}
\safemath{\scl}{l}
\safemath{\scm}{m}
\safemath{\scn}{n}
\safemath{\sco}{o}
\safemath{\scp}{p}
\safemath{\scq}{q}
\safemath{\scr}{r}
\safemath{\scs}{s}
\safemath{\sct}{t}
\safemath{\scu}{u}
\safemath{\scv}{v}
\safemath{\scw}{w}
\safemath{\scx}{x}
\safemath{\scy}{y}
\safemath{\scz}{z}

\safemath{\rsca}{A}
\safemath{\rscb}{B}
\safemath{\rscc}{C}
\safemath{\rscd}{D}
\safemath{\rsce}{E}
\safemath{\rscf}{F}
\safemath{\rscg}{G}
\safemath{\rsch}{H}
\safemath{\rsci}{I}
\safemath{\rscj}{J}
\safemath{\rsck}{K}
\safemath{\rscl}{L}
\safemath{\rscm}{M}
\safemath{\rscn}{N}
\safemath{\rsco}{O}
\safemath{\rscp}{P}
\safemath{\rscq}{Q}
\safemath{\rscr}{R}
\safemath{\rscs}{S}
\safemath{\rsct}{T}
\safemath{\rscu}{U}
\safemath{\rscv}{V}
\safemath{\rscw}{W}
\safemath{\rscx}{X}
\safemath{\rscy}{Y}
\safemath{\rscz}{Z}

\safemath{\veca}{{\bf{a}}}
\safemath{\vecb}{{\bf{b}}}
\safemath{\vecc}{{\bf{c}}}
\safemath{\vecd}{{\bf{d}}}
\safemath{\vece}{{\bf{e}}}
\safemath{\vecf}{{\bf{f}}}
\safemath{\vecg}{{\bf{g}}}
\safemath{\vech}{{\bf{h}}}
\safemath{\veci}{{\bf{i}}}
\safemath{\vecj}{{\bf{j}}}
\safemath{\veck}{{\bf{k}}}
\safemath{\vecl}{{\bf{l}}}
\safemath{\vecm}{{\bf{m}}}
\safemath{\vecn}{{\bf{n}}}
\safemath{\veco}{{\bf{o}}}
\safemath{\vecp}{{\bf{p}}}
\safemath{\vecq}{{\bf{q}}}
\safemath{\vecr}{{\bf{r}}}
\safemath{\vecs}{{\bf{s}}}
\safemath{\vect}{{\bf{t}}}
\safemath{\vecu}{{\bf{u}}}
\safemath{\vecv}{{\bf{v}}}
\safemath{\vecw}{{\bf{w}}}
\safemath{\vecx}{{\bf{x}}}
\safemath{\vecy}{{\bf{y}}}
\safemath{\vecz}{{\bf{z}}}
\safemath{\veczero}{{\bf{0}}}
\safemath{\vecone}{{\bf{1}}}

\safemath{\rveca}{{\bf{A}}}
\safemath{\rvecb}{{\bf{B}}}
\safemath{\rvecc}{{\bf{C}}}
\safemath{\rvecd}{{\bf{D}}}
\safemath{\rvece}{{\bf{E}}}
\safemath{\rvecf}{{\bf{F}}}
\safemath{\rvecg}{{\bf{G}}}
\safemath{\rvech}{{\bf{H}}}
\safemath{\rveci}{{\bf{I}}}
\safemath{\rvecj}{{\bf{J}}}
\safemath{\rveck}{{\bf{K}}}
\safemath{\rvecl}{{\bf{L}}}
\safemath{\rvecm}{{\bf{M}}}
\safemath{\rvecn}{{\bf{N}}}
\safemath{\rveco}{{\bf{O}}}
\safemath{\rvecp}{{\bf{P}}}
\safemath{\rvecq}{{\bf{Q}}}
\safemath{\rvecr}{{\bf{R}}}
\safemath{\rvecs}{{\bf{S}}}
\safemath{\rvect}{{\bf{T}}}
\safemath{\rvecu}{{\bf{U}}}
\safemath{\rvecv}{{\bf{V}}}
\safemath{\rvecw}{{\bf{W}}}
\safemath{\rvecx}{{\bf{X}}}
\safemath{\rvecy}{{\bf{Y}}}
\safemath{\rvecz}{{\bf{Z}}}

\safemath{\matA}{{\bf{A}}}
\safemath{\matB}{{\bf{B}}}
\safemath{\matC}{{\bf{C}}}
\safemath{\matD}{{\bf{D}}}
\safemath{\matE}{{\bf{E}}}
\safemath{\matF}{{\bf{F}}}
\safemath{\matG}{{\bf{G}}}
\safemath{\matH}{{\bf{H}}}
\safemath{\matI}{{\bf{I}}}
\safemath{\matJ}{{\bf{J}}}
\safemath{\matK}{{\bf{K}}}
\safemath{\matL}{{\bf{L}}}
\safemath{\matM}{{\bf{M}}}
\safemath{\matN}{{\bf{N}}}
\safemath{\matO}{{\bf{O}}}
\safemath{\matP}{{\bf{\Pi}}}
\safemath{\matQ}{{\bf{Q}}}
\safemath{\matR}{{\bf{R}}}
\safemath{\matS}{{\bf{S}}}
\safemath{\matT}{{\bf{T}}}
\safemath{\matU}{{\bf{U}}}
\safemath{\matV}{{\bf{V}}}
\safemath{\matW}{{\bf{W}}}
\safemath{\matX}{{\bf{X}}}
\safemath{\matY}{{\bf{Y}}}
\safemath{\matZ}{{\bf{Z}}}
\safemath{\matzero}{{\bf{0}}}

\safemath{\rmatA}{{\bf{A}}}
\safemath{\rmatB}{{\bf{B}}}
\safemath{\rmatC}{{\bf{C}}}
\safemath{\rmatD}{{\bf{D}}}
\safemath{\rmatE}{{\bf{E}}}
\safemath{\rmatF}{{\bf{F}}}
\safemath{\rmatG}{{\bf{G}}}
\safemath{\rmatH}{{\bf{H}}}
\safemath{\rmatI}{{\bf{I}}}
\safemath{\rmatJ}{{\bf{J}}}
\safemath{\rmatK}{{\bf{K}}}
\safemath{\rmatL}{{\bf{L}}}
\safemath{\rmatM}{{\bf{M}}}
\safemath{\rmatN}{{\bf{N}}}
\safemath{\rmatO}{{\bf{O}}}
\safemath{\rmatP}{{\bf{P}}}
\safemath{\rmatQ}{{\bf{Q}}}
\safemath{\rmatR}{{\bf{R}}}
\safemath{\rmatS}{{\bf{S}}}
\safemath{\rmatT}{{\bf{T}}}
\safemath{\rmatU}{{\bf{U}}}
\safemath{\rmatV}{{\bf{V}}}
\safemath{\rmatW}{{\bf{W}}}
\safemath{\rmatX}{{\bf{X}}}
\safemath{\rmatY}{{\bf{Y}}}
\safemath{\rmatZ}{{\bf{Z}}}
\safemath{\rmatzero}{{\bf{0}}}

\safemath{\setA}{\mathcal{A}}
\safemath{\setB}{\mathcal{B}}
\safemath{\setC}{\mathcal{C}}
\safemath{\setD}{\mathcal{D}}
\safemath{\setE}{\mathcal{E}}
\safemath{\setF}{\mathcal{F}}
\safemath{\setG}{\mathcal{G}}
\safemath{\setH}{\mathcal{H}}
\safemath{\setI}{\mathcal{I}}
\safemath{\setJ}{\mathcal{J}}
\safemath{\setK}{\mathcal{K}}
\safemath{\setL}{\mathcal{L}}
\safemath{\setM}{\mathcal{M}}
\safemath{\setN}{\mathcal{N}}
\safemath{\setO}{\mathcal{O}}
\safemath{\setP}{\mathcal{P}}
\safemath{\setQ}{\mathcal{Q}}
\safemath{\setR}{\mathcal{R}}
\safemath{\setS}{\mathcal{S}}
\safemath{\setT}{\mathcal{T}}
\safemath{\setU}{\mathcal{U}}
\safemath{\setV}{\mathcal{V}}
\safemath{\setW}{\mathcal{W}}
\safemath{\setX}{\mathcal{X}}
\safemath{\setY}{\mathcal{Y}}
\safemath{\setZ}{\mathcal{Z}}
\safemath{\emptySet}{\varnothing}

\safemath{\opE}{\mathbb{E}}
\safemath{\opP}{\mathbb{P}}
\safemath{\tp}{T}
\safemath{\herm}{*}
\safemath{\pinv}{\dagger}
\safemath{\tr}{\operatorname{tr}}
\Nsafemath{\det}{\operatorname{det}}
\safemath{\rank}{\operatorname{rank}}
\safemath{\range}{\operatorname{range\,}}
\safemath{\spn}{\operatorname{span}}
\safemath{\spark}{\operatorname{spark}}
\Nsafemath{\Re}{\operatorname{Re}}
\Nsafemath{\Im}{\operatorname{im}}

\safemath{\complex}{\mathbb{C}}
\safemath{\reals}{\mathbb{R}}
\safemath{\integers}{\mathbb{Z}}
\safemath{\naturals}{\mathbb{N}}

\newcommand{\onorm}[1]{{|\hspace{-0.25truemm}|\hspace{-0.25truemm}|#1|\hspace{-0.25truemm}|\hspace{-0.25truemm}|}}

\begin{document}
  



  \def\Sy{s}
\def\Sz{t}
\makeatletter
\newcommand{\subalign}[1]{%
  \vcenter{%
    \Let@ \restore@math@cr \default@tag
    \baselineskip\fontdimen10 \scriptfont\tw@
    \advance\baselineskip\fontdimen12 \scriptfont\tw@
    \lineskip\thr@@\fontdimen8 \scriptfont\thr@@
    \lineskiplimit\lineskip
    \ialign{\hfil$\m@th\scriptstyle##$&$\m@th\scriptstyle{}##$\crcr
      #1\crcr
    }%
  }
}
\makeatother
\allowdisplaybreaks[2]
\chapter{Uncertainty Relations and Sparse Signal Recovery}
\vspace{-21truemm}
{\bf{Erwin Riegler and Helmut B\"olcskei}}

\section{Abstract}

This chapter provides a principled introduction to uncertainty relations\index{subject}{uncertainty relation} underlying sparse signal recovery. We start
with the seminal work by Donoho and Stark, 1989, which defines uncertainty relations\index{subject}{uncertainty relation} as upper bounds on the operator norm of
the band-limitation operator followed by the time-limitation operator, generalize this theory to arbitrary pairs of operators, and then
develop---out of this generalization---the coherence-based\index{subject}{coherence}  uncertainty relations\index{subject}{uncertainty relation} due to Elad and Bruckstein, 2002, as well as uncertainty relations\index{subject}{uncertainty relation} 
in terms of concentration of $1$-norm or $2$-norm. 
The theory is completed with the recently discovered set-theoretic uncertainty relations\index{subject}{uncertainty relation} which lead to best possible recovery thresholds
in terms of a general measure of parsimony, namely Minkowski dimension\index{subject}{Minkowski dimension}. 
We also elaborate on the remarkable connection between uncertainty relations\index{subject}{uncertainty relation} and the ``large sieve"\index{subject}{large sieve}, a family of inequalities developed in analytic number theory. 
It is finally shown how uncertainty relations\index{subject}{uncertainty relation} allow to establish fundamental limits of practical
signal recovery problems such as inpainting, declipping, super-resolution, and denoising
of signals corrupted by impulse noise or narrowband interference. Detailed proofs 
are provided throughout the chapter.

\section{Introduction}

The uncertainty principle in quantum mechanics says that certain pairs of physical properties of a particle, such as position and momentum, can  be known to within a limited precision only \cite{he30}. 
Uncertainty relations\index{subject}{uncertainty relation} in signal analysis  \cite{fa78,copr84,be94,fosi97} state that a signal and its Fourier transform 
can not both be arbitrarily well concentrated; corresponding mathematical formulations exist for  square-integrable or integrable functions 
\cite{dost89,dolo92}, for vectors in $(\complex ^m, \lVert\,\cdot\,\rVert_2)$ or  $(\complex ^m, \lVert\,\cdot\,\rVert_1)$ \cite{dost89,dolo92,elbr02,stkupobo12,kudubo12}, and for finite abelian groups \cite{te99,ta05}. 
These results feature prominently in many areas of the mathematical data sciences. Specifically, in compressed sensing \cite{dost89,dolo92,elbr02,stkupobo12,striagbo17,fora13} uncertainty relations\index{subject}{uncertainty relation} lead to sparse signal recovery thresholds, in Gabor and Wilson frame theory \cite{Groechenig2001} they characterize limits on the time-frequency localization of frame elements, in communications \cite{ga46} they play a fundamental role in the design of pulse shapes for
orthogonal frequency division multiplexing (OFDM) systems \cite{Boel2003}, in the theory of partial differential equations they serve to characterize existence and smoothness properties of solutions \cite{fe83}, and in 
coding theory they help to understand questions around the existence of good cyclic codes \cite{evkolu17}. 


This chapter provides a principled introduction to uncertainty relations\index{subject}{uncertainty relation} underlying sparse signal recovery,
starting with the seminal work by Donoho and Stark \cite{dost89}, ranging over the Elad-Bruckstein 
coherence-based\index{subject}{coherence} uncertainty relation\index{subject}{uncertainty relation} for general pairs of orthonormal bases \cite{elbr02}, later extended to general pairs
of dictionaries \cite{kudubo12}, to the recently discovered set-theoretic uncertainty relation\index{subject}{uncertainty relation} \cite{striagbo17} which leads to information-theoretic recovery thresholds for general
notions of parsimony.  We also elaborate on the remarkable connection \cite{dolo92} between uncertainty relations\index{subject}{uncertainty relation} for signals and their Fourier transforms---with concentration measured 
in terms of support---and the ``large sieve''\index{subject}{large sieve}, a family of inequalities involving trigonometric polynomials, originally
developed in the field of analytic number theory \cite{bo74,mo01}.

Uncertainty relations\index{subject}{uncertainty relation} play an important role in data science beyond sparse signal recovery, specifically in the sparse signal separation\index{subject}{signal separation} problem, 
which comprises numerous practically relevant applications such as (image or audio signal) inpainting, declipping, super-resolution, and the
recovery of signals corrupted by impulse noise or by narrowband interference. We provide a systematic treatment of the sparse signal separation\index{subject}{signal separation} 
problem and develop its limits out of uncertainty relations\index{subject}{uncertainty relation} for general pairs of dictionaries as introduced in \cite{kudubo12}. While the flavor of these
results is that beyond certain thresholds something is not possible, for example a nonzero vector can not be concentrated with respect to two different orthonormal bases 
beyond a certain limit, uncertainty relations\index{subject}{uncertainty relation} can also reveal that something unexpected is possible. 
Specifically, we demonstrate that signals that are sparse in certain bases can be recovered in a stable fashion from partial and noisy observations.

In practice one often encounters more general concepts of parsimony, such as, e.g., manifold structures and fractal sets.
Manifolds are prevalent in the data sciences, e.g., in compressed sensing \cite{bawa09,care09,elkubo10,capl11,albdekori18,ristbo15}, machine learning \cite{lz12}, image processing \cite{lufahe98,soze98}, and handwritten-digit recognition \cite{hidare97}. 
Fractal sets find application in image compression and in modeling of Ethernet traffic \cite{letawi94}. 
In the last part of this chapter, we develop an information-theoretic framework for sparse signal separation\index{subject}{signal separation} and recovery, which applies to 
arbitrary signals of ``low description complexity". The complexity measure our results are formulated in, namely Minkowski dimension\index{subject}{Minkowski dimension}, is agnostic to
signal structure and goes beyond the notion of sparsity\index{subject}{sparsity} in terms
of the number of nonzero entries or concentration in $1$-norm or $2$-norm.
The corresponding recovery thresholds are information-theoretic in the sense of applying to arbitrary signal structures,
and provide results of best possible nature that are, however, not constructive in terms of recovery algorithms.

To keep the exposition  simple and to elucidate the main conceptual aspects, we restrict ourselves to the finite-dimensional cases $(\complex ^m, \lVert\,\cdot\,\rVert_2)$  and $(\complex ^m, \lVert\,\cdot\,\rVert_1)$  throughout. 
References to uncertainty relations\index{subject}{uncertainty relation} for the infinite-dimensional case will be given wherever possible and appropriate. 
Some of the results in this chapter have not been reported before in the literature. 
Detailed proofs will be provided for most of the statements, with the goal of allowing the reader to acquire a technical working knowledge that can serve as a basis for own further research.   

The chapter is organized as follows. 
In Sections \ref{sec:l2} and \ref{sec:l1}, we derive  uncertainty relations\index{subject}{uncertainty relation} for vectors in $(\complex ^m, \lVert\,\cdot\,\rVert_2)$ and  $(\complex ^m, \lVert\,\cdot\,\rVert_1)$, respectively,  discuss the connection to the large sieve\index{subject}{large sieve},  present applications to  noisy signal recovery problems, and 
 establish a fundamental relation between  uncertainty relations\index{subject}{uncertainty relation} for sparse vectors and null-space properties\index{subject}{null-space property} of the accompanying dictionary matrices. 
Section \ref{sec2} is devoted to understanding the role of uncertainty relations\index{subject}{uncertainty relation} in  sparse  signal separation\index{subject}{signal separation} problems. 
In Section \ref{sec3}, we generalize the classical sparsity\index{subject}{sparsity} notion as used in compressed sensing to a more comprehensive concept of description complexity, namely, lower modified Minkowski dimension\index{subject}{Minkowski dimension! modified}, which in turn leads to a set-theoretic null-space property\index{subject}{null-space property! set-theoretic} and corresponding recovery thresholds.   
Section \ref{sec4} presents a large sieve\index{subject}{large sieve} inequality in $(\complex ^m, \lVert\,\cdot\,\rVert_2)$ that one of our results in Section \ref{sec:l2} is based on. 
Section \ref{sec:cont}  lists infinite-dimensional counterparts---available in the literature---to some of the results in this chapter. 
In Section \ref{thm:davidproof}, we provide a proof of the set-theoretic null-space property\index{subject}{null-space property! set-theoretic} stated in Section \ref{sec3}. 
Finally, Section \ref{sec:norm} contains results on operator norms used frequently in this chapter. 

\emph{Notation.} For  $\setA\subseteq\{1,\dots,m\}$, $\matD_\setA$ denotes the $m\times m$ diagonal  matrix with diagonal entries $(\matD_\setA)_{i,i}=1$ for $i\in\setA$, and $(\matD_\setA)_{i,i}=0$ else. 
With $\matU\in\complex^{m\times m}$ unitary and  $\setA\subseteq\{1,\dots,m\}$, we 
define the orthogonal projection  
$\matP_\setA(\matU)=\matU\matD_\setA\matU^\herm$ 
and set $\setW^{\matU,\setA}=\range(\matP_\setA(\matU))$. 
For $\vecx\in\complex^m$ and  $\setA\subseteq\{1,\dots,m\}$, we let $\vecx_\setA=\matD_\setA\vecx$. 
With $\matA\in \complex^{m\times m}$, 
$\onorm{\matA}_1=\max_{\vecx:\, \|\vecx\|_1=1}\|\matA\vecx\|_1$ refers to  the operator $1$-norm, 
$\onorm{\matA}_2=\max_{\vecx:\, \|\vecx\|_2=1}\|\matA\vecx\|_2$ designates the  operator $2$-norm,  
$\|\matA\|_2=\sqrt{\tr(\matA\matA^\herm)}$ is the Frobenius norm, and $\|\matA\|_1=\sum_{i,j=1}^m |A_{i,j}|$. 
The $m\times m$ DFT matrix $\matF$ has entry $(1/\sqrt{m})e^{-2\pi j kl/m}$ in its $k$-th row and $l$-th column for $k,l\in\{1,\dots,m\}$.
For  $x\in\reals$,  we set $[x]_+=\max\{x,0\}$.  
The vector $\vecx\in\complex^m$ is said to be $s$-sparse if it has at most $s$ nonzero entries.  
The open ball in $(\complex^m, \lVert\,\cdot\,\rVert_2)$ of radius $\rho$ centered at $\vecu\in \complex^m$ is denoted by $\setB_m(\vecu,\rho)$ and $V_m(\rho)$  refers to its volume. The indicator function on the set $\setA$ is $\ind{\setA}$. 
We use the convention  $0\cdot\infty=0$. 

\section{Uncertainty Relations\index{subject}{uncertainty relation} in  $(\complex ^m, \lVert\,\cdot\,\rVert_2)$}\label{sec:l2}

Donoho and Stark \cite{dost89} define uncertainty relations\index{subject}{uncertainty relation} as upper bounds on the operator norm of the band-limitation operator followed by the  time-limitation operator. 
We adopt this elegant concept and extend it to refer to an upper bound on the operator norm of a general orthogonal projection operator 
(replacing the band-limitation operator) followed by the ``time-limitation operator" $\matD_\setP$ as an uncertainty relation\index{subject}{uncertainty relation}. More specifically,  
let $\matU\in\complex^{m\times m}$ be a unitary matrix,  $\setP,\setQ\subseteq\{1,\dots,m\}$, and consider the orthogonal projection
$\matP_\setQ(\matU)$ onto the subspace $\setW^{\matU,\setQ}$ which is spanned by  $\{\vecu_i:i\in\setQ\}$. 
Let\footnote{We note that, for general unitary $\matA,\matB\in\complex^{m\times m}$, 
unitary invariance of   $\lVert\,\cdot\,\rVert_2$  yields 
 $\onorm{\matP_\setP(\matA)\matP_\setQ(\matB)}_2=\onorm{\matD_\setP\matP_\setQ(\matU)}_2$ with $\matU=\matA^\herm\matB$. 
 The situation where both the band-limitation and the time-limitation operator are replaced by general orthogonal projection operators can hence be reduced to the case considered here.} 
$\Delta_{\setP,\setQ}(\matU)=\onorm{\matD_\setP\matP_\setQ(\matU)}_2$.
In the setting of \cite{dost89} $\matU$ would correspond to the DFT matrix $\matF$ and 
$\Delta_{\setP,\setQ}(\matF)$ is the operator $2$-norm of the band-limitation operator followed by the  time-limitation operator, both in finite dimensions. 
By Lemma \ref{lem:norm} we have  
\begin{align}\label{eq:Delta2}
 \Delta_{\setP,\setQ}(\matU)=\max_{\vecx\in\setW^{\matU,\setQ}\setminus\{\veczero\}}\frac{\|\vecx_\setP\|_2}{\|\vecx\|_2}.
\end{align}  
An uncertainty relation\index{subject}{uncertainty relation} in $(\complex ^m, \lVert\,\cdot\,\rVert_2)$  is  an upper bound of the form $\Delta_{\setP,\setQ}(\matU)\leq c$ with $c\geq 0$, and states that
 $\|\vecx_\setP\|_2\leq c \|\vecx\|_2$  for all  $\vecx\in\setW^{\matU,\setQ}$.  
$\Delta_{\setP,\setQ}(\matU)$ hence quantifies how well a vector supported on $\setQ$ in the basis $\matU$ can be concentrated on $\setP$. 
Note that an uncertainty relation\index{subject}{uncertainty relation} in $(\complex ^m, \lVert\,\cdot\,\rVert_2)$ is  nontrivial only if $c<1$.   
Application of Lemma  \ref{lem:ineqonorm} now yields 
\begin{align}\label{eq:unl21}
\frac{\|\matD_\setP\matP_\setQ(\matU)\|_2}{\sqrt{\rank(\matD_\setP\matP_\setQ(\matU))}}
\leq
\Delta_{\setP,\setQ}(\matU)
&\leq \|\matD_\setP\matP_\setQ(\matU)\|_2, 
\end{align}
where the upper bound  constitutes an uncertainty relation\index{subject}{uncertainty relation} and 
the lower bound will allow us  to assess its tightness. 
Next, note that
\begin{align}
\|\matD_\setP\matP_\setQ(\matU)\|_2=\sqrt{\tr(\matD_\setP\matP_\setQ(\matU))}
\end{align} 
and 
\begin{align}
\rank(\matD_\setP\matP_\setQ(\matU))
&=\rank(\matD_\setP\matU\matD_\setQ\matU^\herm)\\
&\leq \min\{|\setP|,|\setQ|\}, \label{eq:inequnc}
\end{align} 
where \eqref{eq:inequnc} follows from  $\rank(\matD_\setP\matU\matD_\setQ)\leq \min\{|\setP|,|\setQ|\}$ and \cite[Property (c), Chapter 0.4.5]{hojo13}. 
When used in  \eqref{eq:unl21} this implies 
\begin{align}\label{eq:unl21a}
\sqrt{\frac{\tr(\matD_\setP\matP_\setQ(\matU))}{\min\{|\setP|,|\setQ|\}}}
\leq
\Delta_{\setP,\setQ}(\matU)
&\leq\sqrt{\tr(\matD_\setP\matP_\setQ(\matU))}.
\end{align}
Particularizing to  $\matU=\matF$, we obtain   
\begin{align}
\sqrt{\tr(\matD_\setP\matP_\setQ(\matF))}\label{eq:dost1}
&=\sqrt{\tr(\matD_\setP\matF\matD_\setQ\matF^\herm)}\\
&=\sqrt{\sum_{i\in\setP}\sum_{j\in\setQ}|\matF_{i,j}|^2}\\
&=\sqrt{\frac{|\setP| |\setQ|}{m}},   \label{eq:dost2}
\end{align}
so that \eqref{eq:unl21a} reduces to 
\begin{align}\label{eq:unl21F}
\sqrt{\frac{\max\{|\setP|,|\setQ|\}}{m}}
\leq
\Delta_{\setP,\setQ}(\matF)
&\leq \sqrt{\frac{|\setP| |\setQ|}{m}}. 
\end{align}
There exist sets $\setP,\setQ\subseteq\{1,\dots,m\}$ that saturate both  bounds in \eqref{eq:unl21F}, e.g.,  $\setP=\{1\}$ and $\setQ=\{1,\dots,m\}$, which yields  
$\sqrt{\max\{|\setP|,|\setQ|\}/m}=\sqrt{|\setP| |\setQ|/m}=1$ and therefore $\Delta_{\setP,\setQ}(\matF)=1$.  
An example of sets $\setP,\setQ\subseteq\{1,\dots,m\}$ saturating only the lower bound in \eqref{eq:unl21F} is as follows.  
Take  $n$ to divide $m$ and set
\begin{align} \label{eq:setP}
\setP=\mleft\{\frac{m}{n},\frac{2m}{n},\dots,\frac{(n-1)m}{n},m\mright\}
\end{align}
and
\begin{align}\label{eq:setQ}
\setQ=\{l+1,\dots,l+n\}
\end{align} with 
$l\in\{1,\dots,m\}$ and $\setQ$ interpreted circularly in $\{1,\dots,m\}$.
Then, the upper bound in \eqref{eq:unl21F} is  
\begin{align}\label{eq:exatrace}
\sqrt{\frac{|\setP| |\setQ|}{m}}={\frac{n}{\sqrt{m}}}, 
\end{align}
whereas the lower bound 
becomes 
\begin{align}
\sqrt{\frac{\max\{|\setP|,|\setQ|\}}{m}}=\sqrt{\frac{n}{m}}. 
\end{align}
Thus, for $m\to\infty$  with fixed ratio  $m/n$, 
 the upper bound in \eqref{eq:unl21F} tends to infinity whereas the corresponding lower bound 
 remains constant. 
The following result states that the lower bound in   \eqref{eq:unl21F} is tight for $\setP$ and $\setQ$ as in  \eqref{eq:setP} and \eqref{eq:setQ}, respectively. This implies a lack of tightness of the uncertainty 
relation $\Delta_{\setP,\setQ}(\matF)\leq \sqrt{|\setP| |\setQ|/m}$ by a factor of $\sqrt{n}$.  
The large sieve-based\index{subject}{large sieve} uncertainty relation\index{subject}{uncertainty relation} developed in the next section will be seen to remedy this problem.  

\begin{lemma}\cite[Theorem 11]{dost89}\label{lem:lemex}
Let $n$ divide $m$ and consider  
\begin{align} \label{eq:setP2}
\setP=\mleft\{\frac{m}{n},\frac{2m}{n},\dots,\frac{(n-1)m}{n},m\mright\}
\end{align}
and
\begin{align}\label{eq:setQ2}
\setQ=\{l+1,\dots,l+n\}
\end{align} with 
$l\in\{1,\dots,m\}$ and $\setQ$ interpreted circularly in $\{1,\dots,m\}$.
Then, $\Delta_{\setP,\setQ}(\matF)=\sqrt{n/m}$. 
\end{lemma}
\begin{proof} 
We have   
\begin{align}
\Delta_{\setP,\setQ}(\matF)
&=\onorm{\matP_\setQ(\matF)\matD_\setP}_2\label{eq:stepa1s}\\
&=\onorm{\matD_\setQ\matF^\herm\matD_\setP}_2\label{eq:stepa1sb}\\
&=\max_{\vecx:\,\|\vecx\|_2=1}\|\matD_\setQ\matF^\herm\matD_\setP\vecx\|_2\\
&=\max_{\vecx:\,\vecx\neq\veczero}\frac{\|\matD_\setQ\matF^\herm\vecx_\setP\|_2}{\|\vecx\|_2}\\
&= \max_{ 
\subalign{\vecx:\,\vecx&=\vecx_\setP\\
\phantom{\vecx:\,}\vecx&\neq\veczero
}}
\frac{\|\matD_\setQ\matF^\herm\vecx\|_2}{\|\vecx\|_2},\label{eq:stepa2}
\end{align}
where in \eqref{eq:stepa1s} we applied Lemma \ref{lem:norm} and in 
\eqref{eq:stepa1sb} we used unitary invariance of   $\lVert\,\cdot\,\rVert_2$.  
Next, consider an arbitrary but fixed $\vecx\in\complex^m$ with $\vecx=\vecx_\setP$ and define 
$\vecy\in\complex^n$ according to  $y_s=x_{ms/n}$ for $s=1,\dots,n$. 
It follows that 
\begin{align}
\|\matD_\setQ\matF^\herm\vecx\|^2_2\label{eq:cdcdd1}
&=\frac{1}{{m}}\sum_{q\in\setQ}\Big|\sum_{p\in\setP}x_p\,e^{\frac{2\pi j pq}{m}}\Big|^2\\
&=\frac{1}{{m}}\sum_{q\in\setQ}\Big|\sum_{s=1}^{n}x_{ms/n}\,e^{\frac{2\pi j sq}{n}}\Big|^2\\
&=\frac{1}{{m}}\sum_{q\in\setQ}\Big|\sum_{s=1}^{n}y_{s}\,e^{\frac{2\pi j sq}{n}}\Big|^2\\ 
&=\frac{n}{m}\|\matF^\herm\vecy\|_2^2\label{eq:cdcdd3}\\
&=\frac{n}{m}\|\vecy\|_2^2,\label{eq:cdcdd2}
\end{align}
where $\matF$ in \eqref{eq:cdcdd3} is the $n\times n$ DFT matrix and 
in  \eqref{eq:cdcdd2} we used unitary invariance of $\lVert\,\cdot\,\rVert_2$. With \eqref{eq:cdcdd1}--\eqref{eq:cdcdd2} and 
$\|\vecx\|_2=\|\vecy\|_2$ in \eqref{eq:stepa2}, we get $\Delta_{\setP,\setQ}(\matF)=\sqrt{n/m}$. 
\end{proof}  

\subsection{Uncertainty Relations\index{subject}{uncertainty relation} Based on the Large Sieve\index{subject}{large sieve}}

The uncertainty relation\index{subject}{uncertainty relation} in \eqref{eq:unl21a} is very crude as it simply upper-bounds the operator $2$-norm by the Frobenius norm. 
For $\matU=\matF$ a more sophisticated  upper bound on  $\Delta_{\setP,\setQ}(\matF)$ was reported in \cite[Theorem 12]{dolo92}. 
The proof of this result establishes a remarkable connection to the so-called ``large sieve''\index{subject}{large sieve}, a family of inequalities involving trigonometric polynomials originally  developed in the field of analytic number theory \cite{bo74,mo01}. We next present a slightly improved and generalized version of \cite[Theorem 12]{dolo92}. 

\begin{theorem}\label{thm:nyquist}
Let $\setP\subseteq \{1,\dots,m\}$, $l,n\in\{1,\dots,m\}$, and 
\begin{align}
\setQ=\{l+1,\dots,l+n\}
\end{align} with $\setQ$ interpreted circularly in $\{1,\dots,m\}$. For $\lambda\in(0,m]$, we define the 
circular Nyquist density\index{subject}{circular Nyquist density} $\rho(\setP,\lambda)$ according to 
\begin{align}\label{eq:Nyquist}
\rho(\setP,\lambda)=\frac{1}{\lambda}\max_{r\in[0,m)}|\widetilde{\setP}\cap(r,r+\lambda)|, 
\end{align}
where $\widetilde{\setP}=\setP\cup\{m+p:p\in\setP\}$. 
Then, 
\begin{align}\label{eq:temsieve}
\Delta_{\setP,\setQ}(\matF)\leq \sqrt{\mleft(\frac{\lambda(n-1)}{m}+1\mright) \rho(\setP,\lambda)}
\end{align}
for all $\lambda\in(0,m]$. 
\end{theorem}
\begin{proof}
If $\setP=\emptyset$, then $\Delta_{\setP,\setQ}(\matF)=0$  as a consequence of $\matP_\emptyset(\matF)=\matzero$ and \eqref{eq:temsieve} holds trivially. Suppose  now that  $\setP\neq\emptyset$, consider an arbitrary but fixed $\vecx\in\setW^{\matF,\setQ}$ with $\|\vecx\|_2=1$,  and set $\veca=\matF^\herm\vecx$. 
Then, $\veca=\veca_\setQ$ and, by unitarity of $\matF$, $\|\veca\|_2=1$.  
We have 
\begin{align}
|x_p|^2
&=|(\matF\veca)_p|^2\label{eq:steppoly1}\\
&=\frac{1}{m}\biggl|\sum_{q\in\setQ} a_qe^{-\frac{2\pi jpq}{m}}\biggr|^2\\
&=\frac{1}{m}\biggl|\sum_{k=1}^n a_k e^{-\frac{2\pi jpk}{m}}\biggr|^2\\
&=\frac{1}{m}\mleft|\psi\mleft(\frac{p}{m}\mright)\mright|^2\quad\text{for}\ p\in\{1,\dots,m\},\label{eq:steppoly2} 
\end{align}
where we defined the $1$-periodic trigonometric  polynomial $\psi(s)$ according to 
\begin{align}
\psi(s)=\sum_{k=1}^{n}a_k e^{-2\pi j ks}. 
\end{align}
Next, let $\nu_t$ denote the unit Dirac measure centered at $t\in\reals$ and set $\mu=\sum_{p\in\setP}\nu_{p/m}$ with $1$-periodic extension outside $[0,1)$. 
Then, 
\begin{align}
\|\vecx_\setP\|_2^2
&=\frac{1}{m}\sum_{p\in\setP}\mleft|\psi\mleft(\frac{p}{m}\mright)\mright|^2\label{eq:usepolyexp}\\
&=\frac{1}{m}\int_{[0,1)} |\psi(s)|^2\mathrm d\mu(s) \label{eq:diracS}\\
&\leq  \mleft(\frac{n-1}{m}+\frac{1}{\lambda}\mright)\sup_{r\in[0,1)}\mu\mleft(\mleft(r,r+\frac{\lambda}{m}\mright)\mright)\label{eq:usesieve}
\end{align}
for all $\lambda\in(0,m]$, where 
\eqref{eq:usepolyexp} is by \eqref{eq:steppoly1}--\eqref{eq:steppoly2} and 
in \eqref{eq:usesieve} we applied the large sieve\index{subject}{large sieve} inequality  Lemma \ref{lem:sievel2} with $\delta=\lambda/m$ and $\|\veca\|_2=1$. 
Now, 
\begin{align}
&\sup_{r\in[0,1)} \mu\mleft(\mleft(r,r+\frac{\lambda}{m}\mright)\mright)\label{eq:fffhh0}\\
&=\sup_{r\in[0,m)}\sum_{p\in\setP}(\nu_{p}((r,r+\lambda))+\nu_{m+p}((r,r+\lambda)))\label{eq:fffhh2}\\
&=\max_{r\in[0,m)}|\widetilde{\setP}\cap(r,r+\lambda)|\\
&=\lambda\, \rho(\setP,\lambda)\quad \text{for all}\ \lambda\in(0,m],\label{eq:fffhh1}
\end{align}
where in \eqref{eq:fffhh2} we used the $1$-periodicity of $\mu$. 
Using \eqref{eq:fffhh0}--\eqref{eq:fffhh1} in \eqref{eq:usesieve} yields
\begin{align}
\|\vecx_\setP\|_2^2
&\leq\mleft(\frac{\lambda(n-1)}{m}+1\mright) \rho(\setP,\lambda)\quad \text{for all}\ \lambda\in(0,m].
\end{align}
As $\vecx\in\setW^{\matF,\setQ}$ with $\|\vecx\|_2=1$ was arbitrary, we conclude that 
\begin{align}
\Delta^2_{\setP,\setQ}(\matF)
&=\max_{\vecx\in\setW^{\matF,\setQ}\setminus\{\veczero\}}\frac{\|\vecx_\setP\|^2_2}{\|\vecx\|^2_2}\label{eq:laststep}\\
&\leq\mleft(\frac{\lambda(n-1)}{m}+1\mright) \rho(\setP,\lambda)\quad \text{for all}\ \lambda\in(0,m],
\end{align}
thereby finishing the proof. 
\end{proof}
Theorem \ref{thm:nyquist} slightly improves upon \cite[Theorem 12]{dolo92} by virtue of applying to more general sets $\setQ$ and defining the 
circular Nyquist density\index{subject}{circular Nyquist density} in \eqref{eq:Nyquist} in terms of open intervals $(r,r+\lambda)$. 

We next apply Theorem \ref{thm:nyquist} to specific choices of $\setP$ and $\setQ$.  
First, consider $\setP=\{1\}$ and $\setQ=\{1,\dots,m\}$, 
which were shown to saturate the upper and the lower  bound in \eqref{eq:unl21F}, leading to  $\Delta_{\setP,\setQ}(\matF)=1$. Since $\setP$ consists of a single point,  $\rho(\setP,\lambda)=1/\lambda$ for all $\lambda\in(0,m]$.
Thus, Theorem \ref{thm:nyquist} with $n=m$ yields 
\begin{align}\label{eq:ssss1}
\Delta_{\setP,\setQ}(\matF)\leq \sqrt{\frac{m-1}{m}+\frac{1}{\lambda}}\quad\text{for all}\ \lambda\in (0,m]. 
\end{align}
Setting $\lambda=m$ in \eqref{eq:ssss1}  yields  $\Delta_{\setP,\setQ}(\matF)\leq 1$.

Next, consider $\setP$ and $\setQ$ as in  \eqref{eq:setP} and \eqref{eq:setQ}, respectively, which, as already mentioned, have the uncertainty relation\index{subject}{uncertainty relation} in \eqref{eq:unl21F} lacking tightness by a factor of $\sqrt{n}$.  
Since $\setP$ consists of points  spaced  $m/n$ apart,  we get $\rho(\setP,\lambda)=1/\lambda$ for all $\lambda\in(0,m/n]$.
The upper bound \eqref{eq:temsieve} now becomes 
\begin{align}\label{eq:rrrr1}
\Delta_{\setP,\setQ}(\matF)\leq \sqrt{\frac{n-1}{m}+\frac{1}{\lambda}}\quad\text{for all}\ \lambda\in \mleft(0,\frac{m}{n}\mright]. 
\end{align}
Setting  $\lambda=m/n$ in \eqref{eq:rrrr1} yields 
\begin{align}\label{eq:exasieve}
\Delta_{\setP,\setQ}(\matF)\leq\sqrt{(2n-1)/m}\leq \sqrt{2}\sqrt{n/m}, 
\end{align}
which is tight up to a factor of $\sqrt{2}$ (cf. Lemma \ref{lem:lemex}).  
We hasten to add, however, that the large sieve\index{subject}{large sieve} technique applies to $\matU=\matF$ only. 

\subsection{Coherence-based\index{subject}{coherence} Uncertainty Relation\index{subject}{uncertainty relation}}

We next present an uncertainty relation\index{subject}{uncertainty relation} that is of simple form and applies to general unitary $\matU$.  
To this end, we first introduce the concept of coherence\index{subject}{coherence}  of a matrix. 
\begin{definition}\label{def:coh}
For $\matA=(\veca_1\dots\veca_n)\in\complex^{m\times n}$ with columns $\lVert\,\cdot\,\rVert_2$-normalized to $1$, the  coherence\index{subject}{coherence}  is 
defined as $\mu(\matA)=\max_{i\neq j}|\veca_i^\herm\veca_j|$.
\end{definition}

We  have the following coherence-based\index{subject}{coherence}  uncertainty relation\index{subject}{uncertainty relation} valid for  general unitary $\matU$.  

\begin{lemma}\label{lem:U} 
Let $\matU\in\complex^{m\times m}$ be  unitary  and  $\setP,\setQ\subseteq\{1,\dots,m\}$.  Then, 
\begin{align}\label{eq:normineq}
\Delta_{\setP,\setQ}(\matU)\leq \sqrt{|\setP||\setQ|}\,\mu(\congmat{\matI}{\matU}).  
\end{align}
\end{lemma}
\begin{proof}
The claim follows from 
\begin{align}
\Delta_{\setP,\setQ}^2(\matU)
&\leq\tr(\matD_\setP\matU\matD_\setQ\matU^\herm)\label{eq:ineqnorm2aa}\\
&=\sum_{k\in\setP}\sum_{l\in\setQ} |\matU_{k,l}|^2\\
&\leq |\setP||\setQ|\max_{k, l} |\matU_{k,l}|^2\label{eq:ineqnorm2a}\\
&=|\setP||\setQ|\,\mu^2(\congmat{\matI}{\matU}), \label{eq:ineqnorm2} 
\end{align}
where \eqref{eq:ineqnorm2aa} is by \eqref{eq:unl21a} and in \eqref{eq:ineqnorm2} we used the definition of coherence\index{subject}{coherence}.  
\end{proof}
Since $\mu(\congmat{\matI}{\matF})=1/\sqrt{m}$, Lemma \ref{lem:U} particularized to $\matU=\matF$ recovers the upper bound in   \eqref{eq:unl21F}.  

\subsection{Concentration Inequalities}\label{sec:concentration}

As mentioned at the beginning of this chapter, the classical uncertainty relation\index{subject}{uncertainty relation} in signal analysis quantifies how well concentrated a signal can be in time and frequency. 
In the finite-dimensional setting considered here this amounts to characterizing the concentration of 
$\vecp$ and $\vecq$ in  $\vecp=\matF\vecq$. We will actually study the more general case obtained by replacing 
$\matI$ and $\matF$ by unitary $\matA\in\complex^{m\times m}$ and $\matB\in\complex^{m\times m}$, respectively, and will ask ourselves how well concentrated $\vecp$ and $\vecq$ in $\matA\vecp=\matB\vecq$ can be. 
Rewriting $\matA\vecp=\matB\vecq$ according to  $\vecp=\matU\vecq$ with $\matU=\matA^\herm\matB$, 
we now show how the uncertainty relation \index{subject}{uncertainty relation} in Lemma \ref{lem:U} can be used to answer this question. 
Let us start by introducing a measure for concentration in $(\complex^m,\lVert\,\cdot\,\rVert_2)$. 
 
\begin{definition}
Let $\setP\subseteq\{1,\dots,m\}$ and $\varepsilon_\setP\in[0,1]$. The vector $\vecx\in\complex^m$ is said to be 
$\varepsilon_\setP$-concentrated if $\|\vecx-\vecx_\setP\|_2\leq \varepsilon_\setP\|\vecx\|_2$. 
\end{definition}
The fraction of $2$-norm an $\varepsilon_\setP$-concentrated vector exhibits outside $\setP$ is therefore no more than 
$\varepsilon_\setP$. In particular,  if $\vecx$  is $\varepsilon_\setP$-concentrated  with $\varepsilon_\setP=0$, then $\vecx=\vecx_\setP$ and $\vecx$ is $|\setP|$-sparse. 
The zero vector is trivially $\varepsilon_\setP$-concentrated for all   $\setP\subseteq\{1,\dots,m\}$ and  $\varepsilon_\setP\in[0,1]$. 

We next derive a lower bound on $\Delta_{\setP,\setQ}(\matU)$ for unitary matrices $\matU$ that relate $\varepsilon_\setP$-concentrated vectors $\vecp$ 
to $\varepsilon_\setQ$-concentrated vectors $\vecq$ through $\vecp=\matU\vecq$. The formal statement is as follows.  
\begin{lemma}\label{lem:U1} 
Let $\matU\in\complex^{m\times m}$  be  unitary  and   $\setP,\setQ\subseteq\{1,\dots,m\}$. 
Suppose that   there exist a nonzero $\varepsilon_\setP$-concentrated $\vecp\in\complex^m$ and a 
nonzero $\varepsilon_\setQ$-concentrated $\vecq\in\complex^m$ such that 
$\vecp=\matU\vecq$.    
Then,
\begin{align}
\Delta_{\setP,\setQ}(\matU)\geq [1-\varepsilon_\setP-\varepsilon_\setQ]_+.\label{eq:lemU}
\end{align}
\end{lemma}
\begin{proof} 
We have   
\begin{align}
\|\vecp-\matP_\setQ(\matU)\vecp_\setP\|_2 \label{eq:useaaaa1}
&\leq \|\vecp-\matP_\setQ(\matU)\vecp\|_2 + \|\matP_\setQ(\matU)\vecp_\setP-\matP_\setQ(\matU)\vecp \|_2\\
&\leq \|\vecp-\matP_\setQ(\matU)\vecp\|_2 + \onorm{\matP_\setQ(\matU)}_2\|\vecp_\setP-\vecp \|_2\\
&\leq \|\vecp-\matU\matD_\setQ\matU^\herm\vecp\|_2 + \|\vecp_\setP-\vecp \|_2\\
&= \|\vecq-\vecq_\setQ\|_2 + \|\vecp_\setP-\vecp \|_2,\label{eq:useaaaa3}\\
&\leq \varepsilon_\setQ\|\vecq\|_2+\varepsilon_\setP\|\vecp\|_2\\
&= (\varepsilon_\setP+\varepsilon_\setQ)\|\vecp\|_2,  \label{eq:useaaaa2}
\end{align}
where in \eqref{eq:useaaaa3} we made use of the unitary invariance of  $\lVert\,\cdot\,\rVert_2$. 
It follows that
\begin{align}
\|\matP_\setQ(\matU)\vecp_\setP\|_2
&\geq [\|\vecp\|_2-\|\vecp-\matP_\setQ(\matU)\vecp_\setP\|_2]_+\label{eq:usert}\\
&\geq \|\vecp\|_2[1-\varepsilon_\setP-\varepsilon_\setQ]_+, \label{eq:useaaaa}
\end{align}
where \eqref{eq:usert} is by  the  reverse triangle inequality  and 
in \eqref{eq:useaaaa} we used \eqref{eq:useaaaa1}--\eqref{eq:useaaaa2}. 
Since $\vecp\neq\veczero$ by assumption,  \eqref{eq:usert}--\eqref{eq:useaaaa} implies 
\begin{align}
\mleft\|\matP_\setQ(\matU)\matD_\setP\frac{\vecp}{\|\vecp\|_2}\mright\|_2\geq[1-\varepsilon_\setP-\varepsilon_\setQ]_+,
\end{align}
which  in turn yields 
$\onorm{\matP_\setQ(\matU)\matD_\setP}_2\geq [1-\varepsilon_\setP-\varepsilon_\setQ]_+$. This concludes the proof as  
$\Delta_{\setP,\setQ}(\matU)=\onorm{\matP_\setQ(\matU)\matD_\setP}_2$ by Lemma \ref{lem:norm}.  
\end{proof}

Combining Lemma \ref{lem:U1} with the uncertainty relation\index{subject}{uncertainty relation}  Lemma \ref{lem:U}   yields the announced result stating  that a nonzero vector 
can not  be arbitrarily well  concentrated  with respect to two different orthonormal bases. 

\begin{corollary}\label{cor:U}
Let $\matA,\matB\in\complex^{m\times m}$  be  unitary  and   $\setP,\setQ\subseteq\{1,\dots,m\}$. 
Suppose that  there exist a nonzero $\varepsilon_\setP$-concentrated $\vecp\in\complex^m$ and a 
nonzero $\varepsilon_\setQ$-concentrated $\vecq\in\complex^m$ such that 
$\matA\vecp=\matB\vecq$. 
Then,
\begin{align}
|\setP||\setQ| \geq \frac{[1-\varepsilon_\setP-\varepsilon_\setQ]^2_+}{\mu^2(\congmat{\matA}{\matB})}.
\end{align}
\end{corollary}
\begin{proof}
Let $\matU=\matA^\herm\matB$. Then, by Lemmata \ref{lem:U} and \ref{lem:U1}, we have
\begin{align}\label{eq:zvzv}
[1-\varepsilon_\setP-\varepsilon_\setQ]_+\leq \Delta_{\setP,\setQ}(\matU)\leq \sqrt{|\setP||\setQ|}\,\mu(\congmat{\matI}{\matU}). 
\end{align}
The claim now follows by noting that $\mu(\congmat{\matI}{\matU})=\mu(\congmat{\matA}{\matB})$.
\end{proof}

For  $\varepsilon_\setP=\varepsilon_\setQ=0$, we recover the well-known  Elad-Bruckstein  result. 
\begin{corollary}\cite[Theorem 1]{elbr02}\label{cor:elbr02}
Let $\matA,\matB\in\complex^{m\times m}$  be   unitary. 
If  $\matA\vecp=\matB\vecq$ for nonzero  $\vecp,\vecq\in\complex^m$, then $\|\vecp\|_0\|\vecq\|_0\geq 1/\mu^2(\congmat{\matA}{\matB})$. 
\end{corollary}

\subsection{Noisy Recovery in $(\complex ^m, \lVert\,\cdot\,\rVert_2)$}\label{eq:noisyrecl2}
Uncertainty relations\index{subject}{uncertainty relation} are typically employed to prove that something is not possible. 
For example, by Corollary \ref{cor:U} there is a limit on how well a nonzero vector can be  concentrated with respect to two different orthonormal bases.  
Donoho and Stark \cite{dost89} noticed that uncertainty relations\index{subject}{uncertainty relation} can also be used to show that something unexpected is possible. 
Specifically, \cite[Section 4]{dost89} considers a noisy signal recovery problem, which we now translate to the finite-dimensional setting. 
Let $\vecp,\vecn\in\complex^m$ and  $\setP\subseteq\{1,\dots,m\}$,  set $\setP^c=\{1,\dots,m\}\!\setminus\!\setP$, and  suppose that we observe  $\vecy=\vecp_{\setP^c}+\vecn$.  
Note that the information contained in $\vecp_\setP$ is completely lost in the observation.  
Without   structural assumptions  on $\vecp$, it is therefore not possible to recover information on   $\vecp_\setP$ from $\vecy$. 
However, if  $\vecp$ is sufficiently sparse with respect to an orthonormal basis and $|\setP|$ is sufficiently small, it  turns out that 
all  entries of $\vecp$ can be recovered in a linear fashion to within a  precision determined by the noise level. 
This is often referred to in the literature as stable recovery \cite{dost89}. 
The corresponding formal statement is as follows.  

\begin{lemma}\label{cor:noisy}
Let $\matU\in\complex^{m\times m}$ be unitary,  $\setQ\subseteq\{1,\dots,m\}$,   $\vecp\in\setW^{\matU,\setQ}$, and  
consider  
\begin{align}
\vecy=\vecp_{\setP^c}+\vecn, 
\end{align}
where  $\vecn\in\complex^m$ and $\setP^c=\{1,\dots,m\}\!\setminus\!\setP$ with  $\setP\subseteq\{1,\dots,m\}$. 
If $\Delta_{\setP,\setQ}(\matU)<1$, then there exists a matrix $\matL\in\complex^{m\times m}$  such that 
\begin{align}
\|\matL\vecy-\vecp\|_2\leq C\|\vecn_{\setP^c}\|_2  
\end{align}
with $C=1/(1-\Delta_{\setP,\setQ}(\matU))$.  
In particular, 
\begin{align}\label{eq:assnormQ}
|\setP||\setQ|<\frac{1}{\mu^2(\congmat{\matI}{\matU})}
\end{align}
is sufficient for $\Delta_{\setP,\setQ}(\matU)<1$. 
\end{lemma}
\begin{proof} 
For $\Delta_{\setP,\setQ}(\matU)<1$, it follows that (cf. \cite[p. 301]{hojo13})   
$(\matI-\matD_\setP\matP_\setQ(\matU))$ is invertible with  
\begin{align}
\onorm{(\matI-\matD_\setP\matP_\setQ(\matU))^{-1}}_2
&\leq \frac{1}{1-\onorm{\matD_\setP\matP_\setQ(\matU)}_2}\\
&=\frac{1}{1-\Delta_{\setP,\setQ}(\matU)}.
\end{align}
We now set $\matL=(\matI-\matD_\setP\matP_\setQ(\matU))^{-1}\matD_{\setP^c}$ and note that  
\begin{align}
\matL\vecp_{\setP^c}\label{eq:llll0}
&=(\matI-\matD_\setP\matP_\setQ(\matU))^{-1}\vecp_{\setP^c}\\
&=(\matI-\matD_\setP\matP_\setQ(\matU))^{-1}(\matI-\matD_\setP)\vecp\\
&=(\matI-\matD_\setP\matP_\setQ(\matU))^{-1}(\matI-\matD_\setP\matP_\setQ(\matU))\vecp\label{eq:llll1}\\
&=\vecp, \label{eq:llll2}
\end{align}
where  in \eqref{eq:llll1} we used $\matP_\setQ(\matU)\vecp=\vecp$, which is by assumption.
Next, we  upper-bound  $\|\matL\vecy-\vecp\|_2$ according to  
\begin{align}
\|\matL\vecy-\vecp\|_2
&=\|\matL\vecp_{\setP^c}+ \matL\vecn-\vecp\|_2\\
&=\|\matL\vecn\|_2\label{eq:kkkk}\\
&\leq \onorm{(\matI-\matD_\setP\matP_\setQ(\matU))^{-1}}_2\|\vecn_{\setP^c}\|_2\\
&\leq\frac{1}{1-\Delta_{\setP,\setQ}(\matU)}\|\vecn_{\setP^c}\|_2,
\end{align}
where in \eqref{eq:kkkk} we used \eqref{eq:llll0}--\eqref{eq:llll2}. 
Finally, Lemma  \ref{lem:U}   implies that \eqref{eq:assnormQ} is sufficient for $\Delta_{\setP,\setQ}(\matU)<1$.
\end{proof} 
We next particularize Lemma \ref{cor:noisy} for $\matU=\matF$, 
\begin{align} \label{eq:setP3}
\setP=\mleft\{\frac{m}{n},\frac{2m}{n},\dots,\frac{(n-1)m}{n},m\mright\}  
\end{align}
with $n$ dividing  $m$, 
and 
\begin{align}
\setQ=\{l+1,\dots,l+n\}
\end{align} 
with $l\in \{1,\dots,m\}$ and $\setQ$ interpreted circularly in $\{1,\dots,m\}$.
This means that $\vecp$ is $n$-sparse in $\matF$ and we are missing $n$ entries in the noisy observation $\vecy$.
From Lemma  \ref{lem:lemex} we know that $\Delta_{\setP,\setQ}(\matF)=\sqrt{n/m}$. Since $n$  divides $m$ by assumption, stable recovery of $\vecp$ is possible for $n\leq m/2$. 
In contrast, the coherence-based\index{subject}{coherence}  uncertainty relation in Lemma \ref{lem:U}  yields 
$\Delta_{\setP,\setQ}(\matF) \leq {\frac{n}{\sqrt{m}}}$, and would hence 
suggest that $n^2<m$ is needed for stable recovery.

\section{Uncertainty Relations\index{subject}{uncertainty relation} in $(\complex ^m, \lVert\,\cdot\,\rVert_1)$}\label{sec:l1} 

We introduce uncertainty relations\index{subject}{uncertainty relation} in $(\complex ^m, \lVert\,\cdot\,\rVert_1)$ following the same story line as in 
Section \ref{sec:l2}. Specifically, let $\matU=(\vecu_1\dots\vecu_m)\in\complex^{m\times m}$ be a unitary matrix, $\setP,\setQ\subseteq\{1,\dots,m\}$, and consider the orthogonal projection $\matP_\setQ(\matU)$ 
onto the subspace  $\setW^{\matU,\setQ}$, which is spanned by  $\{\vecu_i:i\in\setQ\}$.  
Let\footnote{
In contrast to the operator $2$-norm, the 
operator $1$-norm is not invariant under unitary transformations so that we do not have  
$\onorm{\matP_\setP(\matA)\matP_\setQ(\matB)}_1= \onorm{\matD_\setP\matP_\setQ(\matA^\herm\matB)}_1$  for general unitary $\matA,\matB$. 
This, however, does not constitute a problem as whenever we apply uncertainty relations\index{subject}{uncertainty relation} in $(\complex ^m, \lVert\,\cdot\,\rVert_1)$, the case of general 
unitary $\matA,\matB$
can always be reduced directly to $\matP_\setP(\matI)=\matD_\setP$ and $\matP_\setQ(\matA^\herm\matB)$, simply by rewriting $\matA\vecp=\matB\vecq$ according to $\vecp=\matA^\herm\matB\vecq$.}
$\Sigma_{\setP,\setQ}(\matU)=\onorm{\matD_\setP\matP_\setQ(\matU)}_1$. 
By Lemma \ref{lem:norm} we have  
\begin{align}\label{eq:Sigma2}
 \Sigma_{\setP,\setQ}(\matU)=\max_{\vecx\in\setW^{\matU,\setQ}\setminus\{\veczero\}}\frac{\|\vecx_\setP\|_1}{\|\vecx\|_1}.
\end{align}  
An uncertainty relation\index{subject}{uncertainty relation} in $(\complex ^m, \lVert\,\cdot\,\rVert_1)$  is  an upper bound of the form  $\Sigma_{\setP,\setQ}(\matU)\leq c$ with $c\geq 0$ and states that
 $\|\vecx_\setP\|_1\leq c \|\vecx\|_1$  for all  $\vecx\in\setW^{\matU,\setQ}$. 
$\Sigma_{\setP,\setQ}(\matU)$ hence quantifies how well a vector supported on $\setQ$ in the basis $\matU$ can be concentrated on $\setP$, where now concentration is measured in terms of $1$-norm.   
Again, an uncertainty relation\index{subject}{uncertainty relation} in $(\complex ^m, \lVert\,\cdot\,\rVert_1)$ is  nontrivial only if $c<1$. 
Application of Lemma  \ref{lem:mmm} yields 
\begin{align}
\frac{1}{m}\|\matD_\setP\matP_\setQ(\matU)\|_1
\leq\Sigma_{\setP,\setQ}(\matU)
\leq \|\matD_\setP\matP_\setQ(\matU)\|_1, \label{eq:up13}  
\end{align}
which constitutes the $1$-norm  equivalent of \eqref{eq:unl21}. 

\subsection{Coherence-based\index{subject}{coherence} Uncertainty Relation\index{subject}{uncertainty relation}}

We next derive a coherence-based\index{subject}{coherence} uncertainty relation\index{subject}{uncertainty relation} for $(\complex ^m, \lVert\,\cdot\,\rVert_1)$, which comes with the same advantages and disadvantages as its $2$-norm counterpart. 


\begin{lemma}\label{lem:U1a} 
Let $\matU\in\complex^{m\times m}$ be a unitary matrix and  $\setP,\setQ\subseteq\{1,\dots,m\}$.  Then, 
\begin{align}\label{eq:normineqa}
\Sigma_{\setP,\setQ}(\matU)\leq |\setP||\setQ|\mu^2(\congmat{\matI}{\matU}).  
\end{align}
\end{lemma}
\begin{proof}
Let $\tilde\vecu_i$ denote the column vectors of $\matU^\herm$. 
It follows from Lemma 
\ref{lem:mmm} that 
\begin{align} \label{eq:stepup3}
\Sigma_{\setP,\setQ}(\matU)=\max_{j\in\{1,\dots,m\}}\|\matD_\setP\matU\matD_\setQ\tilde \vecu_j\|_1. 
\end{align}
With
\begin{align}
\max_{j\in\{1,\dots,m\}}\|\matD_\setP\matU\matD_\setQ\tilde \vecu_j\|_1
&\leq|\setP|\max_{i,j\in\{1,\dots,m\}}|\tilde\vecu_i^\herm\matD_\setQ\tilde \vecu_j|\label{eq:stepup1}\\
&\leq|\setP||\setQ|\max_{i,j,k\in\{1,\dots,m\}}|\matU_{i,k}||\matU_{j,k}|\\
&\leq |\setP||\setQ| \mu^2(\congmat{\matI}{\matU}), \label{eq:stepup2}
\end{align} 
this establishes the proof. 
\end{proof}
For $\setP=\{1\}$, $\setQ=\{1,\dots,m\}$, and $\matU=\matF$, the upper bounds on $\Sigma_{\setP,\setQ}(\matF)$ in  \eqref{eq:up13}  and \eqref{eq:normineqa} coincide and    equal $1$. 
We next present an  example where \eqref{eq:normineqa} is sharper than \eqref{eq:up13}. 
Let $m$ be even, $\setP=\{m\}$, $\setQ=\{1,\dots,m/2\}$, and     $\matU=\matF$. 
Then, \eqref{eq:normineqa} becomes $\Sigma_{\setP,\setQ}(\matF)\leq 1/2$, whereas  
\begin{align}
 \|\matD_\setP\matP_\setQ(\matF)\|_1
 &=\frac{1}{m}\sum_{l=1}^m\mleft|\sum_{k=1}^{m/2}e^{\frac{2\pi j lk}{m}}\mright|\\
 &=\frac{1}{2}+\frac{1}{m}\sum_{l=1}^{m-1}\mleft|\frac{1-e^{\pi jl}}{1-e^{ \frac{2\pi jl}{m}}}\mright|\\
 &=\frac{1}{2}+\frac{2}{m}\sum_{l=1}^{m/2}\frac{1}{\mleft|1-e^{ \frac{2\pi j(2l-1)}{m}}\mright|} \\
 &=\frac{1}{2}+\frac{1}{m}\sum_{l=1}^{m/2}\frac{1}{\sin\mleft(\frac{\pi(2l-1)}{m}\mright)}.  \label{eq:Jensen}  
\end{align}  
Applying  Jensen's inequality \cite[Theorem 2.6.2]{Cover91} to \eqref{eq:Jensen} 
and using $\sum_{l=1}^\frac{m}{2}(2l-1)=(m/2)^2$ 
then yields 
$\|\matD_\setP\matP_\setQ(\mathbf{F})\|_1\geq 1$, which shows that   \eqref{eq:up13} is trivial. 

For $\setP$ and $\setQ$ as in  \eqref{eq:setP} and \eqref{eq:setQ}, respectively, \eqref{eq:normineqa}  becomes 
$\Sigma_{\setP,\setQ}(\matF)\leq n^2/m$, which for fixed ratio $n/m$ increases linearly in $m$ and becomes trivial for $m\geq (m/n)^2$.  
A more sophisticated uncertainty relation\index{subject}{uncertainty relation} 
based on a large sieve\index{subject}{large sieve} inequality exists for strictly band-limited (infinite) $\ell_1$-sequences  \cite[Theorem  14]{dolo92};  a corresponding finite-dimensional result  does not seem  to be available. 

\subsection{Concentration Inequalities}
Analogously to Section \ref{sec:concentration}, we next ask how well concentrated a given signal vector can be in two different orthonormal bases. Here, however, we consider a different measure of concentration accounting for the fact that we deal with the $1$-norm.

\begin{definition}\label{def:concl1}
Let $\setP\subseteq\{1,\dots,m\}$  and $\varepsilon_\setP\in[0,1]$. The vector $\vecx\in\complex^m$ is said to be 
 $\varepsilon_\setP$-concentrated  if $\|\vecx-\vecx_\setP\|_1\leq \varepsilon_\setP\|\vecx\|_1$. 
\end{definition}  
The fraction of $1$-norm an $\varepsilon_\setP$-concentrated vector exhibits outside $\setP$ is therefore no more than $\varepsilon_\setP$. In particular,  if $\vecx$  is $\varepsilon_\setP$-concentrated  for $\varepsilon_\setP=0$, then $\vecx=\vecx_\setP$ and $\vecx$ is $|\setP|$-sparse.
The zero vector is trivially $\varepsilon_\setP$-concentrated for all   $\setP\subseteq\{1,\dots,m\}$ and   $\varepsilon_\setP\in[0,1]$. 
In the remainder of Section \ref{sec:l1},  concentration is with respect to the $1$-norm according to Definition \ref{def:concl1}.

We are now ready to state the announced result on the concentration of a vector in two different orthonormal bases. 

\begin{lemma}\label{lem:appUa}
Let $\matA,\matB\in\complex^{m\times m}$  be  unitary and  $\setP,\setQ\subseteq\{1,\dots,m\}$. 
Suppose that  there exist a nonzero $\varepsilon_\setP$-concentrated $\vecp\in\complex^m$ and a 
nonzero $\vecq\in\complex^m$ with $\vecq=\vecq_\setQ$ such that 
$\matA\vecp=\matB\vecq$. 
Then,
\begin{align}
|\setP||\setQ| \geq \frac{1-\varepsilon_\setP}{\mu^2(\congmat{\matA}{\matB})}.
\end{align}
\end{lemma}
\begin{proof}
Rewriting $\matA\vecp=\matB\vecq$ according to $\vecp=\matA^\herm\matB\vecq$, it follows that 
$\vecp\in\setW^{\matU,\setQ}$ with $\matU=\matA^\herm\matB$. 
We have 
\begin{align}
1-\varepsilon_\setP
&\leq \frac{\|\vecp_\setP\|_1}{\|\vecp\|_1}\label{eq:late1}\\
&\leq \Sigma_{\setP,\setQ}(\matU)\label{eq:late2}\\
&\leq |\setP||\setQ|\mu^2(\congmat{\matI}{\matU}),\label{eq:late3} 
\end{align}
where \eqref{eq:late1} is by $\varepsilon_\setP$-concentration of $\vecp$, 
\eqref{eq:late2} follows from \eqref{eq:Sigma2}  and 
$\vecp\in\setW^{\matU,\setQ}$, 
and in \eqref{eq:late3}  we applied  Lemma \ref{lem:U1a}. The proof is concluded by noting that 
$\mu(\congmat{\matI}{\matU})=\mu(\congmat{\matA}{\matB})$.

\end{proof}
For $\varepsilon_\setP=0$, Lemma \ref{lem:appUa} recovers Corollary \ref{cor:elbr02}. 

\subsection{Noisy Recovery in $(\complex ^m, \lVert\,\cdot\,\rVert_1)$}\label{eq:noisyrecl1}
We next consider a noisy signal recovery problem akin to that in Section \ref{eq:noisyrecl2}. 
Specifically, we investigate recovery---through $1$-norm minimization---of a sparse signal  corrupted by   $\varepsilon_\setP$-concentrated noise.  
\begin{lemma}\label{cor:noisy1}
Let 
\begin{align}
\vecy=\vecp+\vecn, 
\end{align}
where $\vecn\in\complex^m$ is  $\varepsilon_\setP$-concentrated to  
$\setP\subseteq\{1,\dots,m\}$  
and $\vecp\in\setW^{\matU,\setQ}$ for $\matU\in\complex^{m\times m}$ unitary and  $\setQ\subseteq\{1,\dots,m\}$.  
Denote  
\begin{align}
\vecz =\underset{\vecw\in\setW^{\matU,\setQ}}{\operatorname{argmin}}(\|\vecy-\vecw\|_1). 
\end{align} 
If $\Sigma_{\setP,\setQ}(\matU)<1/2$, then 
$\|\vecz-\vecp\|_1\leq C\varepsilon_\setP\|\vecn\|_1$ with $C=2/(1-2\Sigma_{\setP,\setQ}(\matU))$. 
In particular, 
\begin{align}\label{eq:assnormQ1}
|\setP||\setQ|<\frac{1}{2\mu^2(\congmat{\matI}{\matU})}
\end{align}
is sufficient for  $\Sigma_{\setP,\setQ}(\matU)<1/2$. 
\end{lemma}
\begin{proof} 
Set $\setP^c=\{1,\dots,m\}\!\setminus\!\setP$ and let $\vecq=\matU^\herm\vecp$. Note that $\vecq_\setQ=\vecq$ as a consequence of  $\vecp\in\setW^{\matU,\setQ}$, which is by assumption. 
We have 
\begin{align}
\|\vecn\|_1
&=\|\vecy-\vecp\|_1\\
&\geq \|\vecy-\vecz\|_1\\
&=\|\vecn-\tilde \vecz\|_1\label{eq:huu1}\\
&=\|(\vecn-\tilde \vecz)_\setP\|_1+\|(\vecn-\tilde \vecz)_{\setP^c}\|_1\label{eq:huu1a}\\
&\geq \|\vecn_\setP\|_1-\|\vecn_{\setP^c}\|_1 +\|\tilde \vecz_{\setP^c}\|_1 - \|\tilde \vecz_\setP\|_1\label{eq:huu2}\\
&=\|\vecn\|_1-2\|\vecn_{\setP^c}\|_1 +\|\tilde \vecz\|_1 - 2\|\tilde \vecz_\setP\|_1\\
&\geq \|\vecn\|_1(1-2\varepsilon_\setP) +\|\tilde\vecz\|_1\mleft(1-2\Sigma_{\setP,\setQ}(\matU)\mright),\label{eq:huu3} 
\end{align}
where in \eqref{eq:huu1} we set $\tilde \vecz=\vecz-\vecp$,   
in \eqref{eq:huu2} we applied the reverse triangle inequality, and in \eqref{eq:huu3} 
we used that $\vecn$ is $\varepsilon_\setP$-concentrated and $\tilde \vecz\in \setW^{\matU,\setQ}$, owing to  
$\vecz\in \setW^{\matU,\setQ}$ and $\vecp\in \setW^{\matU,\setQ}$, together with \eqref{eq:Sigma2}.   
This  yields 
\begin{align}
\|\vecz-\vecp\|_1
&=\|\tilde \vecz\|_1\\
&\leq \frac{2\varepsilon_\setP}{1-2\Sigma_{\setP,\setQ}(\matU)}\|\vecn\|_1. 
\end{align}
Finally, \eqref{eq:assnormQ1} implies $\Sigma_{\setP,\setQ}(\matU)<1/2$ thanks to   \eqref{eq:normineqa}. 
\end{proof}  
Note that for  $\varepsilon_\setP=0$, i.e., the noise vector is supported on $\setP$, we can 
  recover $\vecp$ from $\vecy=\vecp+\vecn$ perfectly provided that  $\Sigma_{\setP,\setQ}(\matU)<1/2$.   
For the special case  $\matU=\matF$,  this is guaranteed by  
\begin{align}
|\setP||\setQ|<\frac{m}{2},   
\end{align}
and perfect recovery of $\vecp$ from $\vecy=\vecp+\vecn$ amounts to the finite-dimensional version of what is known 
as Logan's phenomenon \cite[Section 6.2]{dost89}. 

\subsection{Coherence-based\index{subject}{coherence}  Uncertainty Relation\index{subject}{uncertainty relation} for Pairs of General Matrices}

In practice, one is often interested in sparse signal representations with respect to general (i.e., possibly redundant or incomplete) dictionaries. The purpose of this section is to provide a corresponding general uncertainty relation\index{subject}{uncertainty relation}. 
Specifically, we consider representations of a given  signal vector $\vecs$ according to $\vecs=\matA\vecp=\matB\vecq$, where $\matA\in\complex^{m\times p}$ and $\matB\in\complex^{m\times q}$ are  general matrices, $\vecp\in\complex^p$, and $\vecq\in\complex^q$. 
We start by introducing the notion of  mutual coherence\index{subject}{coherence! mutual}  for pairs of matrices. 
\begin{definition}\label{def:mutualcoh}
For $\matA=(\veca_1\dots\veca_p)\in\complex^{m\times p}$ and 
$\matB=(\vecb_1\dots\vecb_q)\in\complex^{m\times q}$, both with columns $\lVert\,\cdot\,\rVert_2$-normalized to $1$, 
the  mutual coherence\index{subject}{coherence! mutual} $\bar\mu(\matA,\matB)$   is 
defined as $\bar\mu(\matA,\matB)=\max_{i,j}|\veca_i^\herm\vecb_j|$.
\end{definition}

The general uncertainty relation\index{subject}{uncertainty relation} we are now ready to state is in terms of a pair of upper bounds on $\|\vecp_\setP\|_1$ and $\|\vecq_\setQ\|_1$ for   $\setP\subseteq\{1,\dots,p\}$ and 
$\setQ\subseteq\{1,\dots,q\}$.  

\begin{theorem}\label{thm:stkupobo12}
Let  $\matA\in\complex^{m\times p}$ and $\matB\in\complex^{m\times q}$, both with column vectors $\lVert\,\cdot\,\rVert_2$-normalized to $1$, and consider   $\vecp\in\complex^p$ and $\vecq\in\complex^q$. Suppose that 
$\matA\vecp=\matB\vecq$. Then, we have
\begin{align}
\|\vecp_\setP\|_1&\leq |\setP|\mleft(\frac{\mu(\matA)\|\vecp\|_1+\bar\mu(\matA,\matB)\|\vecq\|_1}{1+\mu(\matA)}\mright)\label{eq:studerfinalAa}
\end{align}
for all  $\setP\subseteq\{1,\dots,p\}$ and, by symmetry,  
\begin{align}
\|\vecq_\setQ\|_1&\leq |\setQ|\mleft(\frac{\mu(\matB)\|\vecq\|_1 + \bar\mu(\matA,\matB)\|\vecp\|_1}{1+\mu(\matB)}\mright)\label{eq:studerfinalAb}
\end{align}
for all  $\setQ\subseteq\{1,\dots,q\}$.  
\end{theorem}
\begin{proof}
Since \eqref{eq:studerfinalAb} follows from \eqref{eq:studerfinalAa} simply by replacing $\matA$ by $\matB$, $\vecp$ by $\vecq$, and 
$\setP$ by $\setQ$, and noting that $\bar\mu(\matA,\matB)=\bar\mu(\matB,\matA)$, it suffices to prove  \eqref{eq:studerfinalAa}. Let  $\setP\subseteq\{1,\dots,p\}$  and consider an arbitrary but fixed $i\in\{1,\dots, p\}$.  Multiplying  $\matA\vecp=\matB\vecq$ from the left by $\veca_i^\herm$ and taking absolute values results in 
\begin{align}
|\veca_i^\herm \matA \vecp|=|\veca_i^\herm \matB \vecq|.\label{eq:eqstuder} 
\end{align}
The left-hand side of  \eqref{eq:eqstuder} can be lower-bounded according to 
\begin{align}
|\veca_i^\herm \matA \vecp|
&= \Big|p_i +\sum_{\subalign{k&=1\\ k&\neq i}}^p\veca_i^\herm\veca_kp_k\Big|\label{eq:lowerbound1}\\
&\geq |p_i|-\Big|\sum_{\subalign{k&=1\\ k&\neq i}}^p\veca_i^\herm\veca_kp_k\Big|\label{eq:usetirrev}\\
&\geq |p_i|-\sum_{\subalign{k&=1\\ k&\neq i}}^p|\veca_i^\herm\veca_k||p_k|\\
&\geq |p_i|-\mu(\matA)\sum_{\subalign{k&=1\\ k&\neq i}}^p|p_k|\label{eq:usemua}\\
&=(1+\mu(\matA))|p_i|-\mu(\matA)\|\vecp\|_1,\label{eq:lowerbound2} 
\end{align}
where \eqref{eq:usetirrev} is by the reverse triangle inequality and in \eqref{eq:usemua} we used Definition \ref{def:coh}. 
Next, we upper-bound the right-hand side of \eqref{eq:eqstuder}  according to  
\begin{align}
|\veca_i^\herm \matB \vecq|
& = \Big|\sum_{k=1}^q\veca_i^\herm\vecb_kq_k\Big|\label{eq:upperbound1}\\
&\leq \sum_{k=1}^q|\veca_i^\herm\vecb_k||q_k|\\
&\leq \bar\mu(\matA,\matB)\|\vecq\|_1,\label{eq:upperbound2}
\end{align}
where the last step is by  Definition \ref{def:mutualcoh}. Combining the lower bound \eqref{eq:lowerbound1}--\eqref{eq:lowerbound2} and the upper bound \eqref{eq:upperbound1}--\eqref{eq:upperbound2} yields 
\begin{align}\label{eq:ineqstuder}
(1+\mu(\matA))|p_i|-\mu(\matA)\|\vecp\|_1\leq \bar\mu(\matA,\matB)\|\vecq\|_1. 
\end{align}
Since \eqref{eq:ineqstuder} holds for arbitrary $i\in\{1,\dots,p\}$, we can sum over all  $i\in\setP$ and get 
\begin{align}
\|\vecp_\setP\|_1\leq|\setP|\mleft(\frac{\mu(\matA)\|\vecp\|_1+\bar\mu(\matA,\matB)\|\vecq\|_1}{1+\mu(\matA)}\mright). \label{eq:PPP}
\end{align}
\end{proof}

For the special case $\matA=\matI\in\complex^{m\times m}$ and $\matB\in\complex^{m\times m}$ with $\matB$ unitary, we have 
$\mu(\matA)=\mu(\matB)=0$ and $\bar\mu(\matI,\matB)=\mu(\congmat{\matI}{\matB})$, so that 
\eqref{eq:studerfinalAa} and \eqref{eq:studerfinalAb} simplify to  
\begin{align}
\|\vecp_\setP\|_1&\leq |\setP|\,\mu(\congmat{\matI}{\matB})\,\|\vecq\|_1\label{eq:studerfinalAUa}
\end{align}
and
\begin{align}
\|\vecq_\setQ\|_1&\leq |\setQ|\,\mu(\congmat{\matI}{\matB})\,\|\vecp\|_1,\label{eq:studerfinalAUb}
\end{align} 
respectively. 
Thus, for arbitrary but fixed $\vecp\in\setW^{\matB,\setQ}$ and $\vecq=\matB^\herm\vecp$, we have $\vecq_\setQ=\vecq$ so that   \eqref{eq:studerfinalAUa} and  \eqref{eq:studerfinalAUb}   taken together yield 
\begin{align}
\|\vecp_\setP\|_1&\leq |\setP||\setQ|\,\mu^2(\congmat{\matI}{\matB})\,\|\vecp\|_1. 
\end{align}
As $\vecp$ was assumed to be  arbitrary,  by \eqref{eq:Sigma2} this recovers the uncertainty relation\index{subject}{uncertainty relation}
\begin{align}
\Sigma_{\setP,\setQ}(\matB)&\leq |\setP||\setQ|\mu^2(\congmat{\matI}{\matB}) \label{eq:normonea}
\end{align}
in Lemma \ref{lem:U1a}. 

\subsection{Concentration Inequalities for Pairs of General Matrices}

We next refine the result in Theorem \ref{thm:stkupobo12} to vectors that are concentrated in $1$-norm according to Definition \ref{def:concl1}. 
 The formal statement is as follows. 

\begin{corollary}\label{cor1:stkupobo12} 
Let  $\matA\in\complex^{m\times p}$ and $\matB\in\complex^{m\times q}$, both with column vectors $\lVert\,\cdot\,\rVert_2$-normalized to $1$, 
$\setP\subseteq\{1,\dots,p\}$, $\setQ\subseteq\{1,\dots,q\}$, $\vecp\in\complex^p$, and $\vecq\in\complex^q$. 
Suppose that $\matA\vecp=\matB\vecq$.  Then, the following statements hold. 
\begin{enumerate}
\item\label{one}
If $\vecq$ is $\varepsilon_\setQ$-concentrated, then, 
\begin{align}
\|\vecp_\setP\|_1\leq  
\frac{|\setP|}{1+\mu(\matA)}\mleft(\mu(\matA)+\frac{\bar\mu^2(\matA,\matB)|\setQ|}{[(1+\mu(\matB))(1-\varepsilon_\setQ)-\mu(\matB)|\setQ|]_+}\mright)\|\vecp\|_1.\label{eq:frame1}
\end{align}
\item\label{two}
If $\vecp$ is $\varepsilon_\setP$-concentrated, then,
\begin{align}
\|\vecq_\setQ\|_1\leq  
\frac{|\setQ|}{1+\mu(\matB)}\mleft(\mu(\matB)+\frac{\bar\mu^2(\matA,\matB)|\setP|}{[(1+\mu(\matA))(1-\varepsilon_\setP)-\mu(\matA)|\setP|]_+}\mright)\|\vecq\|_1. \label{eq:frame2}
\end{align}
\item \label{three}If  $\vecp$ is $\varepsilon_\setP$-concentrated, $\vecq$ is $\varepsilon_\setQ$-concentrated, $\bar\mu(\matA,\matB)>0$, and $(\vecp^\tp\ \vecq^\tp)^\tp \neq\veczero$, then,  
\begin{align}
|\setP||\setQ|\geq 
\frac{[(1+\mu(\matA))(1-\varepsilon_\setP)-\mu(\matA)|\setP|]_+[(1+\mu(\matB))(1-\varepsilon_\setQ)-\mu(\matB)|\setQ|]_+}{\bar\mu^2(\matA,\matB)}.
\label{eq:frame3}
\end{align}
\end{enumerate}
\end{corollary} 
\begin{proof}
By Theorem \ref{thm:stkupobo12}, we have   
\begin{align}
\|\vecp_\setP\|_1&\leq |\setP|\mleft(\frac{\mu(\matA)\|\vecp\|_1+\bar\mu(\matA,\matB)\|\vecq\|_1}{1+\mu(\matA)}\mright)\label{eq:studerfinalBa}
\end{align}
and
\begin{align}
\|\vecq_\setQ\|_1&\leq |\setQ|\mleft(\frac{\mu(\matB)\|\vecq\|_1 + \bar\mu(\matA,\matB)\|\vecp\|_1}{1+\mu(\matB)}\mright).\label{eq:studerfinalBb}
\end{align}
Suppose now that $\vecq$ is $\varepsilon_\setQ$-concentrated, i.e., $\|\vecq_\setQ\|_1\geq (1-\varepsilon_\setQ)\|\vecq\|_1$.  
Then,  \eqref{eq:studerfinalBb}   implies that
\begin{align}\label{eq:dddd}
\|\vecq\|_1\leq \frac{|\setQ|\bar\mu(\matA,\matB)}{[(1+\mu(\matB))(1-\varepsilon_\setQ)-\mu(\matB)|\setQ|]_+}\|\vecp\|_1.
\end{align}
Using \eqref{eq:dddd} in \eqref{eq:studerfinalBa} yields \eqref{eq:frame1}. 
The relation \eqref{eq:frame2} follows from \eqref{eq:frame1} by swapping the roles of $\matA$ and $\matB$, $\vecp$ and $\vecq$, and $\setP$ and $\setQ$, and upon noting that $\bar\mu(\matA,\matB)=\bar\mu(\matB,\matA)$.
It remains to establish \eqref{eq:frame3}. 
Using  $\|\vecp_\setP\|_1\geq (1-\varepsilon_\setP)\|\vecp\|_1$ in \eqref{eq:studerfinalBa} and 
$\|\vecq_\setQ\|_1\geq (1-\varepsilon_\setQ)\|\vecq\|_1$ in \eqref{eq:studerfinalBb} yields
\begin{align}
\|\vecp\|_1[(1+\mu(\matA))(1-\varepsilon_\setP)-\mu(\matA)|\setP|]_+&\leq \bar\mu(\matA,\matB)\|\vecq\|_1|\setP|\label{eq:ineqstuder2}
\end{align}
and
\begin{align}
\|\vecq\|_1[(1+\mu(\matB))(1-\varepsilon_\setQ) -\mu(\matB)|\setQ|]_+&\leq \bar\mu(\matA,\matB)\|\vecp\|_1|\setQ|,\label{eq:ineqstuder2a}
\end{align}
respectively. 
Suppose first that  $\vecp=\veczero$. Then, $\vecq\neq\veczero$ by assumption, and  \eqref{eq:ineqstuder2a}  becomes 
\begin{align}
[(1+\mu(\matB))(1-\varepsilon_\setQ) -\mu(\matB)|\setQ|]_+= 0. 
\end{align}
In this case \eqref{eq:frame3} holds trivially. 
Similarly, if 
$\vecq=\veczero$, then $\vecp\neq\veczero$ again by assumption, and  \eqref{eq:ineqstuder2}  becomes 
\begin{align}
[(1+\mu(\matA))(1-\varepsilon_\setP)-\mu(\matA)|\setP|]_+= 0. 
\end{align}
As before, \eqref{eq:frame3} holds trivially. 
Finally, if $\vecp\neq\veczero$ and $\vecq\neq\veczero$,  then we  multiply \eqref{eq:ineqstuder2} by \eqref{eq:ineqstuder2a} and  divide the result by $\bar\mu^2(\matA,\matB)\|\vecp\|_1\|\vecq\|_1$ which yields \eqref{eq:frame3}. 
\end{proof}
Corollary \ref{cor1:stkupobo12} will be used in Section \ref{sec2} to derive recovery thresholds for sparse signal separation\index{subject}{signal separation}. 
The lower bound  on $|\setP||\setQ|$ in \eqref{eq:frame3} is \cite[Theorem 1]{stkupobo12} 
and states that a nonzero vector can not be arbitrarily well concentrated with respect to two different general matrices $\matA$ and $\matB$.  
For the special case  $\varepsilon_\setQ=0$ and  $\matA$ and $\matB$ unitary, and hence $\mu(\matA)=\mu(\matB)=0$ and 
$\bar\mu(\matA,\matB)=\mu(\congmat{\matA}{\matB})$,   \eqref{eq:frame3} recovers  Lemma \ref{lem:appUa}. 

Particularizing \eqref{eq:frame3} to $\varepsilon_\setP=\varepsilon_\setQ=0$ yields  the following result. 

\begin{corollary}\cite[Lemma 33]{kudubo12}\label{cor:stkupobo12}
Let  $\matA\in\complex^{m\times p}$ and $\matB\in\complex^{m\times q}$, both with column vectors $\lVert\,\cdot\,\rVert_2$-normalized to $1$, and consider  $\vecp\in\complex^p$ and $\vecq\in\complex^q$ with $(\vecp^\tp\ \vecq^\tp)^\tp \neq\veczero$.  
Suppose that $\matA\vecp=\matB\vecq$. 
Then, $\|\vecp\|_0\|\vecq\|_0\geq f_{\matA,\matB}(\|\vecp\|_0, \|\vecq\|_0)$, where 
\begin{align}\label{eq:studerfinal}
f_{\matA,\matB}(u, v)=\frac{[1+\mu(\matA)(1-u)]_+[1+\mu(\matB)(1-v)]_+}{\bar\mu^2(\matA,\matB)}. 
\end{align} 
\end{corollary}
\begin{proof}
Let $\setP=\{i\in\{1,\dots,p\}:p_i\neq 0\}$ and $\setQ=\{i\in\{1,\dots,q\}:q_i\neq 0\}$, so that   
 $\vecp_\setP=\vecp$, $\vecq_\setQ=\vecq$, $|\setP|=\|\vecp\|_0$, and $|\setQ|=\|\vecq\|_0$. 
The claim now follows directly  from \eqref{eq:frame3} with 
$\varepsilon_\setP=\varepsilon_\setQ=0$.
\end{proof}
If  $\matA$ and $\matB$ are both unitary, 
then  $\mu(\matA)=\mu(\matB)=0$ and $\bar\mu(\matA,\matB)=\mu(\congmat{\matA}{\matB})$,  and  Corollary \ref{cor:stkupobo12} recovers the  Elad-Bruckstein result in Corollary \ref{cor:elbr02}. 

Corollary \ref{cor:stkupobo12} admits the following appealing geometric interpretation in terms of a null-space property\index{subject}{null-space property},  
which will be seen in Section \ref{sec3} to pave the way to an extension of the classical notion of sparsity\index{subject}{sparsity} to a more general concept of parsimony.  

 \begin{lemma}\label{lem:ns}
Let  $\matA\in\complex^{m\times p}$ and $\matB\in\complex^{m\times q}$, both with column vectors $\lVert\,\cdot\,\rVert_2$-normalized to $1$. Then,  the set (which actually is a finite union of subspaces) 
\begin{align}\label{eq:usub}
\setS=\Bigg\{
\begin{pmatrix}
\vecp\\
\vecq
\end{pmatrix}:\vecp\in\complex^p,\ \vecq\in\complex^q,\  \|\vecp\|_0\|\vecq\|_0<f_{\matA,\matB}(\|\vecp\|_0, \|\vecq\|_0)\Bigg\} 
\end{align}
with  $f_{\matA,\matB}$ 
defined in \eqref{eq:studerfinal} 
intersects the kernel of $\congmat{\matA}{\matB}$  trivially, i.e., 
\begin{align}
\ker(\congmat{\matA}{\matB})\,\cap\setS=\{\veczero\}.
\end{align} 
\end{lemma}
\begin{proof}
The statement of this lemma is  equivalent to the statement of Corollary \ref{cor:stkupobo12} through a chain of equivalences between the following statements: 
\begin{enumerate}
\item \label{equiv1}
$\ker(\congmat{\matA}{\matB})\cap\setS=\{\veczero\}$;
\item \label{equiv2}
if $(\vecp^\tp -\vecq^\tp)^\tp\in\ker(\congmat{\matA}{\matB})\!\setminus\!\{\veczero\}$, then $\|\vecp\|_0\|\vecq\|_0\geq f_{\matA,\matB}(\|\vecp\|_0, \|\vecq\|_0)$;
\item \label{equiv3} if  $\matA\vecp=\matB\vecq$ with  $(\vecp^\tp\ \vecq^\tp)^\tp \neq\veczero$, then $\|\vecp\|_0\|\vecq\|_0\geq f_{\matA,\matB}(\|\vecp\|_0, \|\vecq\|_0)$, 
\end{enumerate}
where \ref{equiv1} $\Leftrightarrow$  \ref{equiv2} is by the definition of $\setS$, 
\ref{equiv2} $\Leftrightarrow$  \ref{equiv3} follows from the fact that  $\matA\vecp=\matB\vecq$ with  $(\vecp^\tp\ \vecq^\tp)^\tp \neq\veczero$  is equivalent to 
$(\vecp^\tp\ -\vecq^\tp)^\tp\in\ker(\congmat{\matA}{\matB})\!\setminus\!\{\veczero\}$, and \ref{equiv3} is the statement in Corollary \ref{cor:stkupobo12}. 
\end{proof}

\section{Sparse Signal Separation\index{subject}{signal separation}}\label{sec2}

Numerous practical signal recovery tasks can be cast as sparse signal separation\index{subject}{signal separation} problems of the following form. 
We want to recover $\vecy\in\complex^p$  with $\|\vecy\|_0\leq \Sy$ and/or $\vecz\in\complex^q$  with $\|\vecz\|_0\leq \Sz$ from the noiseless observation 
\begin{align}\label{eq:problem1}
\vecw
&=\matA\vecy+\matB\vecz,    
\end{align}
where 
$\matA\in\complex^{m\times p}$ and 
$\matB\in\complex^{m\times q}$. 
Here, $\Sy$ and $\Sz$ are the sparsity levels\index{subject}{sparsity! level} of $\vecy$ and $\vecz$  with corresponding 
ambient dimensions  $p$ and $q$, respectively. 
Prominent applications include (image) inpainting, declipping,  super-resolution, the recovery of signals corrupted by impulse noise, and the separation of (e.g., audio or video) signals into two distinct components \cite[Section I]{stkupobo12}.  We next briefly describe some of these problems. 
\begin{enumerate}
\item \emph{Clipping:} 
Non-linearities in power-amplifiers or in
analog-to-digital converters often cause signal clipping
or saturation \cite{absm91}. 
This effect can be cast into the 
signal model \eqref{eq:problem1}  by setting $\matB=\matI$, 
identifying  $\vecs=\matA\vecy$ with the signal to be clipped, and setting 
$\vecz=(g_a(\vecs)-\vecs)$  
with  $g_a(\cdot)$
realizing entry-wise clipping 
of the amplitude 
to the interval $[0,a]$.  
If the clipping level $a$ is not too small, then $\vecz$ will be  sparse, i.e., $t\ll q$. 

\item \emph{Missing entries:} 
Our framework also encompasses super-resolution  \cite{mayu10,elhe01} and inpainting \cite{besacaba00} of, e.g., images, audio, and video signals. 
In both these applications only a subset of the entries of the (full-resolution) signal vector
$\vecs = \matA\vecy$ is available and the task is to fill in the
missing entries, which are  accounted for by writing $\vecw=\vecs+\vecz$ with  
$z_i=-s_i$ if the $i$-th entry of $\vecs$ is missing and  $z_i=0$ else. If the number of entries missing is not too large, then $\vecz$ is sparse, i.e., 
$t\ll q$.

\item \emph{Signal separation\index{subject}{signal separation}:} 
Separation of (audio, image, or video) signals
into two structurally distinct components also fits into the  framework described above.
A prominent example  is the separation of
texture from cartoon parts in images (see \cite{elstqudo05,doku13} 
and references therein). The
matrices $\matA$ and $\matB$ are chosen to allow
for sparse representations of the two distinct features.
Note that here  $\matB\vecz$
no longer plays the role of undesired noise and the goal 
is to recover both $\vecy$ and $\vecz$ from the observation  $\vecw=\matA\vecy+\matB\vecz$. 
\end{enumerate}

The first two examples above demonstrate that in many practically relevant applications the locations of the 
possibly nonzero entries of one of the sparse vectors, say $\vecz$, may be known.  
This can be accounted for by removing the columns of $\matB$ corresponding to the other entries, 
which results in  $t=q$, i.e., the sparsity level\index{subject}{sparsity! level} of $\vecz$ equals the ambient dimension. 
We next show how Corollary \ref{cor1:stkupobo12} can be used to state a sufficient condition  for recovery of $\vecy$  from $\vecw=\matA\vecy+\matB\vecz$ 
when  $t=q$. For recovery guarantees in the case where  the sparsity levels\index{subject}{sparsity! level} of both $\vecy$ and $\vecz$ are strictly smaller than their corresponding ambient dimensions, we refer  to \cite[Theorem 8]{stkupobo12}. 

\begin{theorem}\cite[Theorem 4, Theorem 7]{stkupobo12}\label{thm:studl1}
Let $\vecy\in\complex^p$ with $\|\vecy\|_0\leq \Sy$, 
$\vecz\in\complex^q$,  
$\matA\in\complex^{m\times p}$, and $\matB\in\complex^{m\times q}$, 
both with column vectors $\lVert\,\cdot\,\rVert_2$-normalized to $1$ and $\bar\mu(\matA,\matB)>0$. 
Suppose that    
\begin{align}\label{eq:threshsigsepa}
 2\Sy q < f_{\matA,\matB}( 2\Sy, q)
\end{align}
with 
\begin{align}\label{eq:studerfinal2}
f_{\matA,\matB}(u, v)=\frac{[1+\mu(\matA)(1-u)]_+[1+\mu(\matB)(1-v)]_+}{\bar\mu^2(\matA,\matB)}. 
\end{align}
Then, $\vecy$ can be recovered from $\vecw=\matA\vecy+\matB\vecz$ by either of the following algorithms: 
\begin{align}
&(\mathrm{P0})\ \ 
\begin{cases}
\text{minimize}\ \|\tilde\vecy\|_0\\
\text{subject to}\ \matA\tilde\vecy\in\{\vecw+\matB\tilde\vecz: \tilde\vecz\in\complex^q\}. 
\end{cases}\label{eq:P0}\\[2mm]
&(\mathrm{P1})\ \ 
\begin{cases}
\text{minimize}\ \|\tilde\vecy\|_1\\
\text{subject to}\ \matA\tilde\vecy\in\{\vecw+\matB\tilde\vecz: \tilde\vecz\in\complex^q\}. 
\end{cases}\label{eq:P1}
\end{align}
\end{theorem}
\begin{proof}
We provide the proof for  $(\text{P1})$ only.   
The proof for recovery through $(\text{P0})$  is very similar and can be found in 
\cite[Appendix B]{stkupobo12}.  

Let $\vecw=\matA\vecy+\matB\vecz$ and suppose  that  
$(\text{P1})$ delivers $\tilde\vecy\in\complex^p$. This implies $\|\tilde \vecy\|_1\leq \|\vecy\|_1$
and the existence of  a $\tilde\vecz\in\complex^q$   
such that   
\begin{align}
\matA\tilde\vecy=\vecw+\matB\tilde\vecz. \label{eq:sigsep2a}
\end{align}
On the other hand, we also have 
\begin{align}
\matA\vecy=\vecw-\matB\vecz.\label{eq:sigsep1a}
\end{align}
Subtracting \eqref{eq:sigsep1a} from \eqref{eq:sigsep2a}  yields 
\begin{align}
\matA(\underbrace{\tilde\vecy-\vecy}_{=\vecp})=\matB(\underbrace{\tilde\vecz+\vecz}_{=\vecq}).
\end{align}
We now set  
\begin{align}\label{eq:setU}
\setU&=\{i\in\{1,\dots,p\}:y_i\neq 0\}
\end{align}
and
\begin{align}
\setU^c&=\{1,\dots,p\}\!\setminus\!\setU
\end{align} 
and show that  $\vecp$ is $\varepsilon_\setU$-concentrated (with respect to $1$-norm) for $\varepsilon_\setU=1/2$, i.e., 
\begin{align}
\|\vecp_{\setU^c}\|_1&\leq \frac{1}{2}\|\vecp\|_1.\label{eq:conc1}  
\end{align}
We have
\begin{align}
\| \vecy\|_1
&\geq \|\tilde \vecy\|_1\label{eq:aaa1}\\
&=\|\vecy+\vecp\|_1\\
&=\|\vecy_\setU+\vecp_\setU\|_1+\|\vecp_{\setU^c}\|_1\label{eq:useV}\\
&\geq \|\vecy_\setU\|_1-\|\vecp_\setU\|_1+\|\vecp_{\setU^c}\|_1\label{eq:steptrirev}\\
&= \|\vecy\|_1-\|\vecp_\setU\|_1+\|\vecp_{\setU^c}\|_1,  \label{eq:aaa2} 
\end{align} 
where \eqref{eq:useV} follows from the definition of $\setU$ in \eqref{eq:setU}, and 
in \eqref{eq:steptrirev} we applied the reverse triangle inequality. 
Now, \eqref{eq:aaa1}--\eqref{eq:aaa2} implies  
$\|\vecp_\setU\|_1\geq\|\vecp_{\setU^c}\|_1$. Thus, 
$2\|\vecp_{\setU^c}\|_1\leq\|\vecp_\setU\|_1+\|\vecp_{\setU^c}\|_1=\|\vecp\|_1$, which establishes \eqref{eq:conc1}. 
Next, set $\setV=\{1,\dots,q\}$ and note that  $\vecq$ is trivially  $\varepsilon_\setV$-concentrated (with respect to  $1$-norm) for $\varepsilon_\setV=0$. 
Suppose, toward a contradiction, that $\vecp\neq\veczero$. Then, we have  
\begin{align}
&2 \Sy q\label{eq:bbb1}\\
&\geq2|\setU||\setV|\\
&\geq 
\frac{[(1+\mu(\matA))-2\mu(\matA)|\setU|]_+[1+\mu(\matB)(1-|\setV|)]_+}{\bar\mu^2(\matA,\matB)}\label{eq:usethmstud} \\
&\geq 
\frac{[(1+\mu(\matA))-2\Sy\mu(\matA)]_+[1+\mu(\matB)(1-q)]_+}{\bar\mu^2(\matA,\matB)},  \label{eq:bbb2} 
\end{align} 
where  \eqref{eq:usethmstud} is obtained by applying Part \ref{three} of 
Corollary \ref{cor1:stkupobo12} with $\vecp$  $\varepsilon_\setU$-concentra\-ted for $\varepsilon_\setU=1/2$ and $\vecq$ 
$\varepsilon_\setV$-concentrated for $\varepsilon_\setV=0$.  
But \eqref{eq:bbb1}---\eqref{eq:bbb2} contradicts \eqref{eq:threshsigsepa}. Hence, we  
must have $\vecp=\veczero$, which yields  $\tilde\vecy=\vecy$. 
\end{proof}


We next provide an example showing that, as soon as \eqref{eq:threshsigsepa} is saturated, recovery through  $(\text{P0})$ or $(\text{P1})$ can fail.  
Take $m=n^2$ with $n$ even,  $\matA=\matF\in\complex^{m\times m}$, 
and $\matB\in\complex^{m\times \sqrt{m}}$ containing every $\sqrt{m}$-th column of the   $m\times m$ identity matrix, i.e., 
\begin{align}\label{eq:matB}
B_{k,l}=\begin{cases}
1&\quad\text{if}\ k=\sqrt{m}\,l\\
0&\quad\text{else}
\end{cases}
\end{align}
for all $k\in\{1,\dots,m\}$ and $l\in\{1,\dots,\sqrt{m}\}$.   
For every $a\in\naturals$ dividing $m$, we define the vector $\vecd^{(a)}\in\complex^m$ with components 
\begin{align}
d^{(a)}_l=
\begin{cases}
1&\quad\text{if}\ l\in\mleft\{a,2a,\dots,\mleft(\frac{m}{a}-1\mright)a,m\mright\}\\
0&\quad\text{else}.
\end{cases}
\end{align} 
Straightforward calculations now yield   
\begin{align}\label{eq:Fd}
\matF\vecd^{(a)}=\frac{\sqrt{m}}{a}\vecd^{(m/a)}
\end{align}
for all $a\in\naturals$  dividing $m$. Suppose that $\vecw=\matF\vecy+\matB\vecz$  
with
\begin{align}
\vecy&=\vecd^{(2\sqrt{m})}-\vecd^{(\sqrt{m})}\in\complex^{m}\label{eq:defy}\\
\vecz&=(1\dots 1)^\tp\in\complex^{\sqrt{m}}.\label{eq:defz} 
\end{align}
Evaluating \eqref{eq:threshsigsepa} for $\matA=\matF$, $\matB$ as defined in \eqref{eq:matB}, and $q=\sqrt{m}$ results in  $s<\sqrt{m}/2$. 
Now, $\vecy$ in \eqref{eq:defy} has $\|\vecy\|_0=\sqrt{m}/2$  and thus just violates the threshold $s<\sqrt{m}/2$. 
We next show that this slender violation is enough for the existence of an alternative  pair $\tilde\vecy\in\complex^{m}$, $\tilde\vecz\in\complex^{\sqrt{m}}$ satisfying   $\vecw=\matF\tilde\vecy+\matB\tilde\vecz$ with  $\|\tilde\vecy\|_0=\|\vecy\|_0$ and  $\|\tilde\vecy\|_1=\|\vecy\|_1$.  
Thus,  neither $(\text{P0})$  nor  $(\text{P1})$ can distinguish between $\vecy$ and $\tilde\vecy$. 
Specifically, we set  
\begin{align}
\tilde \vecy&=\vecd^{(2\sqrt{m})}\in\complex^{m}\\
\tilde \vecz&=\veczero\in\complex^{m} 
\end{align}
and note that   $\|\tilde \vecy\|_0=\|\vecy\|_0=\|\tilde \vecy\|_1=\|\vecy\|_1=\sqrt{m}/2$. It remains to establish that 
$\vecw=\matF\tilde\vecy+\matB\tilde\vecz$. 
To this end,  first note that  
\begin{align}
\vecw
&=\matF\vecy+\matB\vecz\\
&=\frac{1}{2}\vecd^{(\sqrt{m}/2)}-\vecd^{(\sqrt{m})}+\matB\vecz\label{eq:stepsss}\\
&=\frac{1}{2}\vecd^{(\sqrt{m}/2)},\label{eq:stepttt}
\end{align}
where \eqref{eq:stepsss} follows from \eqref{eq:Fd}  and  \eqref{eq:stepttt} is by \eqref{eq:matB}. 
Finally,  again using \eqref{eq:Fd}, we find that
\begin{align}
\matF\tilde\vecy+\matB\tilde\vecz
&=\frac{1}{2}\vecd^{(\sqrt{m}/2)},
\end{align}
which completes the argument. 

The threshold  $s<\sqrt{m}/2$ constitutes a special instance of the so-called  ``square-root bottleneck''\index{subject}{square-root bottleneck} \cite{tr08a} all 
cohe\-rence-based deterministic recovery thresholds suffer from.  
The square-root bottleneck\index{subject}{square-root bottleneck} says that the number of measurements, $m$, has to scale at least quadratically in the sparsity level\index{subject}{sparsity! level} $s$. It
can be circumvented by considering random models for either the signals or the measurement matrices  
\cite{carota06, pobrst13, do06, tr08}, leading to thresholds of the form $m\propto s\log p$ and applying with high probability. 
Deterministic linear recovery thresholds, i.e., $m \propto s$, were, to the best of our knowledge, first been reported in \cite{dojohost92} for the DFT measurement matrix under positivity constraints on the vector to be recovered. Further instances of deterministic linear recovery thresholds were discovered in 
the context of spectrum-blind sampling \cite{febr96,elmo10} and system identification \cite{hebo13}.

\section{The Set-Theoretic Null-Space Property\index{subject}{null-space property! set-theoretic}}\label{sec3} 
The notion of sparsity\index{subject}{sparsity} underlying the theory developed so far is that of either the number of nonzero entries 
or of concentration in terms of $1$-norm or $2$-norm. 
In practice, one often encounters more general concepts of parsimony, such as manifold  or fractal set structures. 
Manifolds are prevalent in data science, e.g., in compressed sensing \cite{bawa09,care09,elkubo10,capl11,albdekori18,ristbo15}, machine learning \cite{lz12}, image processing \cite{lufahe98,soze98}, and handwritten-digit recognition \cite{hidare97}. 
Fractal sets find application in image compression and in modeling of Ethernet traffic \cite{letawi94}. 
Based on the null-space property\index{subject}{null-space property} established in Lemma \ref{lem:ns}, 
we now extend the theory to account for more general notions of parsimony.  
To this end, we first need a suitable measure of ``description complexity'' that goes beyond the concepts of sparsity\index{subject}{sparsity} and concentration. 
Formalizing this idea requires an adequate dimension measure, which, as it turns out,  is lower modified 
Minkowski dimension\index{subject}{Minkowski dimension! modified}.
We start by defining  Minkowski dimension  and modified Minkowski dimension. 

\begin{definition}\cite[Section 3.1]{fa90}\footnote{Minkowski dimension is sometimes also referred to as box-counting dimension, which is the origin of the subscript B in the notation $\dim_\mathrm{B}(\cdot)$ used henceforth.} \label{definitiondim}
For $\setU\subseteq \complex^{m}$ nonempty, the lower and upper Minkowski dimensions\index{subject}{Minkowski dimension} of $\setU$ are defined as 
\begin{align}
\underline{\dim}_\mathrm{B}(\setU)&=\liminf_{\rho\to 0} \frac{\log N_\setU(\rho)}{\log \frac{1}{\rho}}\label{eq:lowerB}
\end{align}
and
\begin{align}
\overline{\dim}_\mathrm{B}(\setU)&=\limsup_{\rho\to 0} \frac{\log N_\setU(\rho)}{\log \frac{1}{\rho}},\label{eq:upperB}
\end{align} 
respectively, where 
\begin{align}
N_\setU(\rho)=\min\Big\{k \in\naturals : \setU\subseteq\hspace*{-4truemm} \bigcup_{i\in\{1,\dots,k\}}\hspace*{-4truemm} \setB_{m}(\vecu_i,\rho),\ \vecu_i\in \setU\Big\}\label{eq:coveringnumber}
\end{align}
is the covering number of $\setU$ for radius $\rho>0$.  
If $\underline{\dim}_\mathrm{B}(\setU)=\overline{\dim}_\mathrm{B}(\setU)$, this common value, denoted by  $\dim_\mathrm{B}(\setU)$, is   
the Minkowski dimension of $\setU$.
\end{definition}

\begin{definition}\cite[Section 3.3]{fa90} \label{definitiondimlocal} 
For $\setU\subseteq \complex^{m}$ nonempty, the lower and upper modified Minkowski dimensions\index{subject}{Minkowski dimension! modified} of $\setU$ are defined as 
\begin{align}
\underline{\dim}_\mathrm{MB}(\setU)=\inf\mleft\{\sup_{i\in\naturals} \underline{\dim}_\mathrm{B}(\setU_i) : \setU\subseteq \bigcup_{i\in\naturals}\setU_i\mright\}
\end{align}
and
\begin{align}
\overline{\dim}_\mathrm{MB}(\setU)=\inf\mleft\{\sup_{i\in\naturals} \overline{\dim}_\mathrm{B}(\setU_i) : \setU\subseteq \bigcup_{i\in\naturals}\setU_i\mright\},
\end{align}
respectively, where in both cases the infimum  is over all possible  coverings $\{\setU_i\}_{i\in\naturals}$ of $\setU$ by nonempty compact sets $\setU_i$. 
If $\underline{\dim}_\mathrm{MB}(\setU)=\overline{\dim}_\mathrm{MB}(\setU)$, this common value, denoted by  $\dim_\mathrm{MB}(\setU)$, is   
the modified Minkowski dimension\index{subject}{Minkowski dimension! modified} of $\setU$.
\end{definition}
For further details on (modified) Minkowski dimension\index{subject}{Minkowski dimension! modified}, we refer the interested reader to  \cite[Section 3]{fa90}. 

We are now ready to extend the null-space property\index{subject}{null-space property} in Lemma \ref{lem:ns} to the following set-theoretic null-space property\index{subject}{null-space property! set-theoretic}.


\begin{theorem}\label{thm:david}
Let $\setU\subseteq\complex^{p+q}$ be nonempty  with  $\underline{\dim}_\mathrm{MB}(\setU)<2m$, and let $\matB\in\complex^{m\times q}$ with $m\geq q$ be a full-rank matrix. Then, 
$\ker \congmat{\matA}{\matB}\cap(\setU\!\setminus\!\{\veczero\})=\emptyset$ for Lebesgue a.a. $\matA\in\complex^{m\times p}$. 
\end{theorem}
\begin{proof}
See Section \ref{thm:davidproof}. \end{proof}
The set $\setU$ in this set-theoretic null-space property\index{subject}{null-space property! set-theoretic} generalizes  the finite union of linear subspaces $\setS$ in  Lemma \ref{lem:ns}. 
For $\setU\subseteq\reals^{p+q}$, the equivalent of Theorem \ref{thm:david}  was reported previously in   \cite[Proposition 1]{striagbo17}.  
The set-theoretic null-space property\index{subject}{null-space property! set-theoretic} can be interpreted in geometric terms as follows.
If $p+q\leq m$, then $\congmat{\matA}{\matB}$ is a tall matrix so that 
the kernel of $\congmat{\matA}{\matB}$ is  $\{\veczero\}$ for Lebesgue-a.a. matrices $\matA$. The statement of the theorem holds trivially in this case.  
If $p+q> m$, then the kernel of $\congmat{\matA}{\matB}$ is a $(p+q-m)$-dimensional  subspace of the ambient space $\complex^{p+q}$ for Lebesgue-a.a. matrices $\matA$. 
The  theorem therefore says that, for Lebesgue-a.a.  $\matA$, the set $\setU$ intersects the subspace $\ker(\congmat{\matA}{\matB})$   at most trivially if 
the sum of $\dim\ker(\congmat{\matA}{\matB})$ and\footnote{The factor $1/2$ stems from the fact that  
(modified) Minkowski dimension\index{subject}{Minkowski dimension! modified}  ``counts real dimensions''. For example, the modified Minkowski dimension\index{subject}{Minkowski dimension! modified} of an $n$-dimensional linear subspace of $\complex^m$ is $2n$ \cite[Example II.2]{albdekori18}.
} $\underline{\dim}_\mathrm{MB}(\setU)/2$  is 
strictly smaller than the  dimension of the ambient space. 
What is remarkable here is that the 
notions of Euclidean dimension (for the kernel of $\congmat{\matA}{\matB}$) and of lower modified
Minkowski dimension\index{subject}{Minkowski dimension! modified} (for the set $\setU$) are compatible. We finally note that, by virtue of the chain of
equivalences in the proof of Lemma \ref{lem:ns}, the set-theoretic null-space property\index{subject}{null-space property! set-theoretic} in Theorem \ref{thm:david} 
leads to a set-theoretic uncertainty relation\index{subject}{uncertainty relation}, albeit not in the form of an upper bound on an operator norm; for a detailed discussion of
this equivalence the interested reader is referred to \cite{striagbo17}.

We next put the set-theoretic null-space property\index{subject}{null-space property! set-theoretic} in Theorem \ref{thm:david} in perspective with the 
null-space property\index{subject}{null-space property} in Lemma \ref{lem:ns}. Fix the sparsity levels\index{subject}{sparsity! level} $s$ and $t$, consider 
the set
\begin{align}\label{eq:usub2}
\setS_{s,t}=\Bigg\{
\begin{pmatrix}
\vecp\\
\vecq
\end{pmatrix}:\vecp\in\complex^p, \vecq\in\complex^q,  \|\vecp\|_0\leq s, \|\vecq\|_0\leq t\Bigg\},   
\end{align}
which is  a finite union of  
 $(s+t)$-dimensional linear subspaces, and, for the sake of concreteness, let $\matA=\matI$ and $\matB=\matF$ of size $q \times q$.
Lemma  \ref{lem:ns} then states that the kernel of $\congmat{\matI}{\matF}$  intersects $\setS_{s,t}$ trivially provided that 
\begin{align}
m>st,\label{eq:npq}
\end{align}
which leads to a recovery threshold in the signal separation problem\index{subject}{signal separation} that is quadratic in the sparsity levels\index{subject}{sparsity! level} $s$ and $t$ \cite[Theorem 8]{stkupobo12}. 
To see what the set-theoretic null-space property\index{subject}{null-space property! set-theoretic} gives, we start by noting that, by 
\cite[Example II.2]{albdekori18}, $\dim_\mathrm{MB}(\setS_{s,t})=2(s+t)$.  
Theorem \ref{thm:david} hence states that, for Lebesgue a.a. matrices $\matA\in\complex^{m\times p}$, 
the kernel of  $\congmat{\matA}{\matB}$ intersects $\setS_{s,t}$ trivially, provided that
\begin{align}
m>s+t. \label{eq:spt}
\end{align}
This is striking as it says that, while the threshold in \eqref{eq:npq} is quadratic in the sparsity levels\index{subject}{sparsity! level} $s$ and $t$ and, therefore, suffers from the square-root bottleneck\index{subject}{square-root bottleneck}, the threshold in \eqref{eq:spt} is linear in $s$ and $t$. 

To understand the operational implications of the observation just made, we demonstrate how the set-theoretic null-space property\index{subject}{null-space property! set-theoretic} 
 in Theorem \ref{thm:david}  leads to  a sufficient condition  for   the recovery of vectors in sets of small lower modified Minkowski dimension\index{subject}{Minkowski dimension! modified}.  

\begin{lemma} \label{lem:appnull-spaceprob}
Let $\setS\subseteq\complex^{p+q}$ be nonempty with $\underline{\dim}_\mathrm{MB}(\setS\ominus\setS)< 2m$, where $\setS\ominus\setS=\{\vecu-\vecv:\vecu,\vecv\in\setS\}$, and let $\matB\in\complex^{m\times q}$, with $m\geq q$, be a full-rank matrix. Then, $\congmat{\matA}{\matB}$ is one-to-one on $\setS$ for Lebesgue a.a. $\matA\in\complex^{m\times p}$. 
\end{lemma} 
\begin{proof}
Follows from the set-theoretic null-space property\index{subject}{null-space property! set-theoretic} in Theorem \ref{thm:david} and the linearity of 
matrix-vector multiplication. 
\end{proof}

To elucidate the implications of Lemma \ref{lem:appnull-spaceprob}, consider  $\setS_{s,t}$ defined in \eqref{eq:usub2}. 
Since $\setS_{s,t}\ominus\setS_{s,t}$ is again a finite union of linear subspaces of dimensions no larger than $\min\{p,2s\}+\min\{q,2t\}$,  where the $\min\{\cdot,\cdot\}$-operation accounts for the fact that the dimension of a linear subspace can not exceed  the dimension of its ambient space, 
we have \cite[Example II.2]{albdekori18} 
\begin{align}
{\dim}_\mathrm{MB}(\setS_{s,t}\ominus\setS_{s,t})=
2(\min\{p,2s\}+\min\{q,2t\}).  
\end{align}
Application of Lemma \ref{lem:appnull-spaceprob} now yields that, for Lebesgue a.a. matrices $\matA\in\complex^{m\times p}$, we can recover  
$\vecy\in\complex^p$ with $\|\vecy\|_0\leq s$ and 
$\vecz\in\complex^q$ with $\|\vecz\|_0\leq t$ 
from $\vecw=\matA\vecy+\matB\vecz$ 
provided that $m>\min\{p,2s\}+\min\{q,2t\}$. 
This qualitative behavior (namely, linear in $s$ and $t$)  is best possible as it  can not be improved even if the support sets of 
 $\vecy$ and $\vecz$ were  known prior to recovery. 
We emphasize, however, that the statement in Lemma \ref{lem:appnull-spaceprob} guarantees injectivity of $\congmat{\matA}{\matB}$ only absent computational considerations for recovery.


\section{A Large Sieve\index{subject}{large sieve} Inequality in $(\complex ^m, \lVert\,\cdot\,\rVert_2)$} \label{sec4}
We  present a slightly improved and generalized version of the large sieve\index{subject}{large sieve} inequality stated in \cite[Equation (32)]{dolo92}.

\begin{lemma}\label{lem:sievel2}
Let $\mu$ be a $1$-periodic, $\sigma$-finite measure on $\reals$, $n\in\naturals$, $\varphi\in[0,1)$, $\veca\in\complex^n$, and consider the 
$1$-periodic trigonometric polynomial 
\begin{align}
\psi(s)=e^{2\pi j\varphi}\sum_{k=1}^{n}a_k e^{-2\pi j ks}. 
\end{align}
Then, 
\begin{align}\label{eq:toshowa}
\int_{[0,1)} |\psi(s)|^2 \mathrm d\mu(s) \leq  \mleft(n-1+\frac{1}{\delta}\mright)\sup_{r\in[0,1)}\mu((r,r+\delta)) \|\veca\|_2^2
\end{align}
for all $\delta\in (0,1]$. 
\end{lemma}
\begin{proof} 
Since 
\begin{align}
|\psi(s)|=\mleft|\sum_{k=1}^{n}a_k e^{-2\pi j ks}\mright|, 
\end{align}
we can assume, without loss of generality, that $\varphi=0$. 
The proof now follows closely the line of argumentation in  \cite[pp. 185--186]{va85} and in the proof of \cite[Lemma 5]{dolo92}. 
Specifically, we make use of the result in  
 \cite[p. 185]{va85} saying that, for every $\delta>0$, there exists a function $g\in L^2(\reals)$ with  Fourier transform 
\begin{align}
G(s)=\int_{-\infty}^\infty g(t) e^{-2\pi j st} \mathrm d t
\end{align}
such that 	
$\|G\|_2^2 =n-1+1/\delta$, $|g(t)|^2\geq 1$ for all $t\in[1,n]$, 
  and  $G(s)=0$ for all $s\notin [-\delta/2,\delta/2]$. 
With this  $g$, consider the $1$-periodic trigonometric polynomial  
\begin{align} 
\theta(s)= \sum_{k=1}^{n}\frac{a_k}{g(k)} e^{-2\pi j ks}     
\end{align}
and note that 
\begin{align}
\int_{-\delta/2}^{\delta/2} G(r) \theta(s-r) \mathrm d r\label{eq:www1}
&=\sum_{k=1}^{n}\frac{a_k}{g(k)}e^{-2\pi j ks} \int_{-\infty}^\infty G(r) e^{2\pi j kr} \mathrm d r\\
&=\sum_{k=1}^{n}a_ke^{-2\pi j ks}\\ 
&=\psi(s)\quad \text{for all}\ s\in\reals.   \label{eq:www2} 
\end{align}
We now have 
\begin{align}
\int_{[0,1)} \mathrm |\psi(s)|^2\mathrm d \mu(s)
&=\int_{[0,1)}  \Bigg|\int_{-\delta/2}^{\delta/2}G(r) \theta(s-r) \mathrm d r\Bigg|^2 \mathrm d \mu(s)\label{eq:vvv2}\\
&\leq \|G\|_2^2 \int_{[0,1)} \mleft( \int_{-\delta/2}^{\delta/2}|\theta(s-r)|^2\mathrm d r\mright)  \mathrm d \mu(s)\label{eq:vvv3}\\
&= \|G\|_2^2\int_{[0,1)} \mleft(  \int_{s-\delta/2}^{s+\delta/2}|\theta(r)|^2\mathrm d r\mright)  \mathrm d \mu(s)\\
&= \|G\|_2^2\int_{-1}^{2}\mu\big((r-\delta/2,r+\delta/2)\cap[0,1)\big) |\theta(r)|^2\mathrm d r\label{eq:vvv4}\\
&= \|G\|_2^2\sum_{i=-1}^1 \int_{0+i}^{1+i}\mu\big((r-\delta/2,r+\delta/2)\cap[0,1)\big) |\theta(r)|^2\mathrm d r\\
&= \|G\|_2^2\sum_{i=-1}^1 \int_{0}^{1}\mu\big((r-\delta/2,r+\delta/2)\cap[i,1+i)\big) |\theta(r)|^2\mathrm d r\label{eq:vvv4a1}\\
&= \|G\|_2^2\int_{0}^{1}\mu\big((r-\delta/2,r+\delta/2)\cap[-1,2)\big) |\theta(r)|^2\mathrm d r\label{eq:vvv4a}\\
&= \|G\|_2^2\int_{0}^{1}\mu\big((r-\delta/2,r+\delta/2)\big) |\theta(r)|^2\mathrm d r\label{eq:vvv5}
\end{align}
for all $\delta\in (0,1]$,  where \eqref{eq:vvv2} follows from  \eqref{eq:www1}--\eqref{eq:www2}, 
in \eqref{eq:vvv3} we applied the Cauchy-Schwartz inequality \cite[Theorem 1.37]{he11}, 
\eqref{eq:vvv4} is by Fubini's theorem \cite[Theorem 1.14]{ma99} (recall that $\mu$ is $\sigma$-finite by assumption) upon noting that   
\begin{align}
&\{(r,s):s\in[0,1),r\in(s-\delta/2,s+\delta/2)\}\\
&=\{(r,s):r\in [-1,2), s\in(r-\delta/2,r+\delta/2)\cap[0,1)\}
\end{align}
for all  $\delta\in(0,1]$, 
in \eqref{eq:vvv4a1}  we used the $1$-periodicity of $\mu$ and $\theta$,  and \eqref{eq:vvv4a} is by  $\sigma$-additivity of $\mu$. 
Now, 
\begin{align}
 \int_{0}^1 \mu\big((r-\delta/2,r+\delta/2)\big) |\theta(r)|^2\mathrm d r\label{eq:kkkk1}
 &\leq \sup_{r\in [0,1)}\mu((r,r+\delta)) \int_0^1 |\theta(r)|^2\mathrm d r\\
 &=\sup_{r\in [0,1)}\mu((r,r+\delta)) \sum_{k=1}^{n}\frac{|a_k|^2}{|g(k)|^2}\\
&\leq  \sup_{r\in [0,1)}\mu((r,r+\delta))\|\veca\|_2^2 \label{eq:kkkk2}
\end{align}
for all $\delta>0$, 
where \eqref{eq:kkkk2} follows from $|g(t)|^2\geq 1$ for all $t\in[1,n]$. Using \eqref{eq:kkkk1}--\eqref{eq:kkkk2} and $\|G\|^2_2=n-1+1/\delta$  in \eqref{eq:vvv5} establishes \eqref{eq:toshowa}.  
\end{proof}

Lemma \ref{lem:sievel2} is a slightly strengthened version of the large sieve inequality\index{subject}{large sieve} \cite[Equation (32)]{dolo92}. 
Specifically, in \eqref{eq:toshowa}  it is sufficient to consider open intervals $(r,r+\delta)$, whereas \cite[Equation (32)]{dolo92} requires    closed intervals $[r,r+\delta]$. 
Thus, the upper bound in \cite[Equation (32)]{dolo92} can be strictly larger than that in  \eqref{eq:toshowa} whenever  $\mu$ has mass points.  

\section{Uncertainty Relations\index{subject}{uncertainty relation} in $L_1$ and $L_2$}\label{sec:cont}
The following table contains a list of  infinite-dimensional counterparts---available in the literature---to results in this chapter. 
Specifically, these results apply to  band-limited $L_1$- and $L_2$-functions and hence  correspond to $\matA=\matI$  and 
$\matB=\matF$ in our setting. 

\begin{center}
\begin{longtable}{|l|l|l|}
\hline
  & $L_2$ analog & $L_1$ analog\\\hline
 Upper bound in  \eqref{eq:unl21F} & \cite[Lemma 2]{dost89}&\\\hline
 Corollary  \ref{cor:U}  & \cite[Theorem 2]{dost89}&\\\hline
 Lemma \ref{cor:noisy}&\cite[Theorem 4]{dost89}&\\\hline
 Lemma \ref{lem:U1a}&&\cite[Lemma 3]{dost89}\\\hline
 Lemma \ref{cor:noisy1}&&\cite[Lemma 2]{dolo92}\\\hline
 Lemma \ref{lem:sievel2}&&\cite[Theorem 4]{dolo92}\\\hline
\end{longtable}
\end{center}


\section{Proof of Theorem \ref{thm:david}}\label{thm:davidproof}\label{sec:proof}
By definition of lower modified Minkowski dimension, \index{subject}{Minkowski dimension! modified} there exists a covering $\{\setU_i\}_{i\in\naturals}$ of $\setU$ 
by nonempty compact sets $\setU_i$ satisfying  $\underline{\dim}_\mathrm{B}(\setU_i)<2m$ for all $i\in\naturals$. The countable subadditivity of Lebesgue measure $\lambda$ now implies 
\begin{align}
&\lambda(\{\matA\in\complex^{m\times p}:\ker \congmat{\matA}{\matB}\cap(\setU\!\setminus\!\{\veczero\})\neq\emptyset\})\\
&\leq\sum_{i=1}^\infty \lambda(\{\matA\in\complex^{m\times p}:\ker \congmat{\matA}{\matB}\cap(\setU_i\!\setminus\!\{\veczero\})\neq\emptyset\}). \label{eq:sumzero1}
\end{align}
We next establish that every term in the sum on the right-hand side  of \eqref{eq:sumzero1} equals zero. 
Take an arbitrary but fixed $i\in\naturals$. 
Repeating the steps in \cite[Equation (10)--(14)]{striagbo17} shows that it suffices  to prove that
\begin{align}\label{eq:toshow}
\opP[\ker(\congmat{\rmatA}{\matB})\cap\setV\neq\emptyset]=0
\end{align}
with 
\begin{align}
\setV=\Bigg\{\begin{pmatrix}\vecu\\\vecv\end{pmatrix}:\vecu\in\complex^p, \vecv\in\complex^q, \|\vecu\|_2>0\Bigg\}\cap\setU_i 
\end{align}
and  $\rmatA=({\rveca}_1\dots{\rveca}_m)^\herm$, where the random vectors ${\rveca}_i$  are independent and  uniformly  distributed  on $\setB_p(\veczero,r)$ for arbitrary but fixed $r>0$. Suppose, toward a contradiction, that \eqref{eq:toshow} is false. This implies 
\begin{align}
0&=\liminf_{\rho\to 0}\frac{\log \opP[\ker(\congmat{\rmatA}{\matB})\cap\setV\neq\emptyset]}{\log\frac{1}{\rho}}\\
 &\leq \liminf_{\rho\to 0}\frac{\log \sum_{i=1}^{N_\setV(\rho)}\opP[\ker(\congmat{\rmatA}{\matB})\cap\setB_{p+q}(\vecc_i,\rho) \neq\emptyset]}{\log\frac{1}{\rho}}, \label{eq:ns3}
\end{align}
where  we have chosen $\{\vecc_i:i=1,\dots,N_\setV(\rho)\}\subseteq\setV$ such that 
\begin{align}
\setV\subseteq \bigcup_{i=1}^{N_\setV(\rho)}\setB_{p+q}(\vecc_i,\rho) 
\end{align}
with $N_\setV(\rho)$ denoting the covering number of $\setV$ for radius $\rho>0$ (cf. \eqref{eq:coveringnumber}). 
Now  let $i\in\{1,\dots,N_\setV(\rho)\}$ be arbitrary but fixed and write $\vecc_i=(\vecu_i^\tp\ \vecv_i^\tp)^\tp$.   
It follows that 
\begin{align}
\|\rmatA\vecu_i+\matB\vecv_i\|_2
&=\|\congmat{\rmatA}{\matB}\vecc_i\|_2\label{eq:nsa1}\\
&\leq \|\congmat{\rmatA}{\matB}(\vecx-\vecc_i)\|_2 +\|\congmat{\rmatA}{\matB}\vecx\|_2\\
&\leq \|\congmat{\rmatA}{\matB}\|_2\|(\vecx-\vecc_i)\|_2 +\|\congmat{\rmatA}{\matB}\vecx\|_2\\
&\leq (\|[\rmatA\|_2+\|\matB\|_2)\rho +\|\congmat{\rmatA}{\matB}\vecx\|_2\\
&\leq (r\sqrt{m}+\|\matB\|_2)\rho +\|\congmat{\rmatA}{\matB}\vecx\|_2\quad\text{for all}\ \vecx\in\setB_{p+q}(\vecc_i,\rho),\label{eq:nsa2}
\end{align}
where in the last step we made use of  $\|\rveca_i\|_2\leq r$ for $i=1,\dots,m$. 
We now have  
\begin{align}
&\opP[\ker(\congmat{\rmatA}{\matB})\cap\setB_{p+q}(\vecc_i,\rho) \neq\emptyset]\label{eq:ns1}\\
&\leq\opP[\exists\,\vecx\in\setB_{p+q}(\vecc_i,\rho): \|\congmat{\rmatA}{\matB}\vecx\|_2<\rho]\\
&\leq \opP[\|\rmatA\vecu_i+\matB\vecv_i\|_2<\rho(1+r\sqrt{m}+\|\matB\|_2)]\label{eq:ns3a}\\
&\leq \rho^{2m} \frac{C(p,m,r)}{\|\vecu_i\|_2^{2m}}(1+r\sqrt{m}+\|\matB\|_2)^{2m},\label{eq:ns2}
\end{align}
where \eqref{eq:ns3a} is by \eqref{eq:nsa1}--\eqref{eq:nsa2}, and in \eqref{eq:ns2} 
we applied the concentration of measure result Lemma \ref{lem:com} below (recall that $\vecc_i=(\vecu_i^\tp\ \vecv_i^\tp)^\tp\in\setV$ implies $\vecu_i\neq\veczero$) with $C(p,m,r)$ as in \eqref{eq:Cpmr} below. 
Inserting \eqref{eq:ns1}--\eqref{eq:ns2} into \eqref{eq:ns3} yields 
\begin{align}
0&\leq \liminf_{\rho\to\infty} \frac{\log(N_\setV(\rho)\rho^{2m})}{\log\frac{1}{\rho}}\\
 &=\underline{\dim}_\mathrm{B}(\setV)-2m\\
 &<0,\label{eq:mb1}
\end{align}
where \eqref{eq:mb1} follows from $\underline{\dim}_\mathrm{B}(\setV)\leq \underline{\dim}_\mathrm{B}(\setU_i)<2m$, 
which constitutes a contradiction. Therefore, \eqref{eq:toshow} must hold. 
\qed
\begin{lemma}\label{lem:com}
Let $\rmatA=({\rveca}_1\dots{\rveca}_m)^\herm$ with independent random vectors ${\rveca}_i$ uniformly distributed on $\setB_p(\veczero,r)$ for  $r>0$. Then, 
\begin{align}
\opP[\|\rmatA\vecu+\vecv\|_2<\delta]\leq \frac{C(p,m,r)}{\|\vecu\|_2^{2m}}\delta^{2m}
\end{align}  
with
\begin{align}\label{eq:Cpmr}
C(p,m,r)=\mleft(\frac{\pi V_{p-1}(r)}{V_{p}(r)}\mright)^m
\end{align}
for all $\vecu\in\complex^p\!\setminus\!\{\veczero\}$,  $\vecv\in\complex^m$, and  $\delta>0$.
\end{lemma}
\begin{proof}
Since 
\begin{align}
\opP[\|\rmatA\vecu+\vecv\|_2<\delta]\label{eq:com1}
&\leq\prod_{i=1}^m\opP[|{\rveca}_i^\herm\vecu+v_i|<\delta]
\end{align}
owing to the independence of the $\rmatA_i$  and as $\|\rmatA\vecu+\vecv\|_2<\delta$ implies $|{\rveca}_i^\herm\vecu+v_i|<\delta$ for  $i=1,\dots,m$, it is sufficient to show that 
\begin{align}
\opP[|\rvecb^\ast\vecu+v|<\delta]\leq \frac{D(p,r)}{\|\vecu\|_2^2}\delta^2
\end{align} 
for all $\vecu\in\complex^p\!\setminus\!\{\veczero\}$,  $v\in\complex$, and  $\delta>0$, where the random vector $\rvecb$ is  uniformly distributed on $\setB_p(\veczero,r)$ and 
\begin{align}
D(p,r)=\frac{\pi V_{p-1}(r)}{V_{p}(r)}.
\end{align} 
We have 
\begin{align}
\opP[|\rvecb^\herm\vecu+v|<\delta]
&=\opP\Bigg[\frac{|\rvecb^\herm\vecu+v |}{\|\vecu\|_2}<\frac{\delta}{\|\vecu\|_2}\Bigg]\\
&=\opP[|\rvecb^\herm\matU^\herm\vece_1+\tilde v|<\tilde \delta]\label{eq:come1}\\
&=\opP[|\rvecb^\herm\vece_1+\tilde v|<\tilde \delta]\label{eq:come2}\\
&=\frac{1}{V_{p}(r)}\int_{\setB_p(\veczero,r)} \ind{\{b_1: |b_1+\tilde v|\,<\, \tilde\delta\}}(b_1)\mathrm d \vecb\label{eq:come3}\\
&\leq \frac{1}{V_{p}(r)}\int_{|b_1+\tilde v|\leq \tilde\delta}\mathrm d b_1\ \int_{\setB_{p-1}(\veczero,r)}\mathrm d (b_2\dots b_p)^\tp\\
&=\frac{V_{p-1}(r)}{V_{p}(r)}\int_{|b_1+\tilde v|<\tilde\delta} \mathrm d b_1\\
&=\frac{V_{p-1}(r)}{V_{p}(r)}\pi{\tilde\delta}^2\label{eq:come4}\\
&=\frac{\pi V_{p-1}(r)}{V_{p}(r)\|\vecu\|^2}{\delta}^2, 
\end{align}
where the unitary matrix $\matU$ in \eqref{eq:come1} has been chosen such that   $\matU(\vecu/\|\vecu\|_2)=\vece_1=(1\ 0\ \dots\ 0)^\tp\in\complex^p$ and we set  $\tilde\delta:=\delta/\|\vecu\|_2$ and  $\tilde v:=v/\|\vecu\|_2$. 
Further,  \eqref{eq:come2} follows from unitary invariance of the uniform distribution on $\setB_p(\veczero,r)$, 
and in \eqref{eq:come3} the factor $1/V_{p}(r)$ is owing to the assumption of a uniform probability density function on $\setB_p(\veczero,r)$.  
\end{proof}

\section{Results for $\onorm{\cdot}_1$ and $\onorm{\cdot}_2$}\label{sec:norm}

\begin{lemma}\label{lem:norm}
Let $\matU\in\complex^{m\times m}$ be unitary, $\setP, \setQ\subseteq\{1,\dots,m\}$, and  consider 
the orthogonal projection $\matP_\setQ(\matU)=\matU\matD_\setQ\matU^\herm$ 
onto the subspace $\setW^{\matU,\setQ}$.  
Then, 
\begin{align}
\onorm{\matP_\setQ(\matU)\matD_\setP}_2
&=\onorm{\matD_\setP\matP_\setQ(\matU)}_2\label{eq:normid1}. 
\end{align} 
Moreover, we have  
\begin{align}\label{eq:norm2A}
 \onorm{\matD_\setP\matP_\setQ(\matU)}_2=\max_{\vecx\in\setW^{\matU,\setQ}\setminus\{\veczero\}}\frac{\|\vecx_\setP\|_2}{\|\vecx\|_2}
\end{align}
and
\begin{align} \label{eq:norm1A}
 \onorm{\matD_\setP\matP_\setQ(\matU)}_1&=\max_{\vecx\in\setW^{\matU,\setQ}\setminus\{\veczero\}}\frac{\|\vecx_\setP\|_1}{\|\vecx\|_1}. 
\end{align}
\end{lemma}
\begin{proof}
The identity \eqref{eq:normid1} follows from 
\begin{align}
\onorm{\matD_\setP\matP_\setQ(\matU)}_2
&=\onorm{(\matD_\setP\matP_\setQ(\matU))^\ast}_2\label{eq:selfedj}\\
&=\onorm{\matP_\setQ^\ast(\matU)\matD_\setP^\ast}_2\\
&=\onorm{\matP_\setQ(\matU)\matD_\setP}_2,    
\end{align}
where in \eqref{eq:selfedj} we used that $\onorm{\cdot}_2$ is self-adjoint \cite[p. 309]{hojo13}, $\matP_\setQ(\matU)^\ast=\matP_\setQ(\matU)$, and $\matD_\setP^\ast=\matD_\setP$. 
To establish \eqref{eq:norm2A}, we note that  
\begin{align}
\onorm{\matD_\setP\matP_\setQ(\matU)}_2
&=\max_{\vecx:\,\|\vecx\|_2=1}\|\matD_\setP\matP_\setQ(\matU)\vecx\|_2\label{eq:stepa1}\\
&= \max_{ 
\substack{\vecx:\, \matP_\setQ(\matU)\vecx\neq\veczero\\
\phantom{\vecx:\,}\|\vecx\|_2=1
}}
\mleft\|\matD_\setP\matP_\setQ(\matU)\vecx\mright\|_2\\
&\leq \max_{ 
\substack{\vecx:\,\matP_\setQ(\matU)\vecx\neq\veczero\\
\phantom{\vecx:\,}\|\vecx\|_2=1
}}
\mleft\|\matD_\setP\frac{\matP_\setQ(\matU)\vecx}{\|\matP_\setQ(\matU)\vecx\|_2}\mright\|_2\label{eq:norm1}\\
&\leq \max_{\vecx:\, \matP_\setQ(\matU)\vecx\neq\veczero}
\mleft\|\matD_\setP\frac{\matP_\setQ(\matU)\vecx}{\|\matP_\setQ(\matU)\vecx\|_2}\mright\|_2\\
&=\max_{\vecx\in\setW^{\matU,\setQ}\setminus\{\veczero\}}\frac{\|\vecx_\setP\|_2}{\|\vecx\|_2}\\
&=\max_{\vecx:\,\matP_\setQ(\matU)\vecx\neq\veczero}\mleft\|\matD_\setP\matP_\setQ(\matU)\frac{\matP_\setQ(\matU)\vecx}{\|\matP_\setQ(\matU)\vecx\|_2}\mright\|_2\\
&\leq \max_{\vecx:\, \|\vecx\|_2=1}\|\matD_\setP\matP_\setQ(\matU)\vecx\|_2\\
&=\onorm{\matD_\setP\matP_\setQ(\matU)}_2,\label{eq:stepa2x}
\end{align}
where in \eqref{eq:norm1} we used $\|\matP_\setQ(\matU)\vecx\|_2\leq \|\vecx\|_2$, which implies $\|\matP_\setQ(\matU)\vecx\|_2\leq 1$ for all $\vecx$ with $\|\vecx\|_2=1$. 
Finally, \eqref{eq:norm1A} follows by repeating the steps in \eqref{eq:stepa1}--\eqref{eq:stepa2x} with $\lVert\,\cdot\,\rVert_2$ replaced by $\lVert\,\cdot\,\rVert_1$ at all occurrences. 
\end{proof} 

\begin{lemma}\label{lem:ineqonorm}
Let $\matA\in\complex^{m\times n}$. Then, 
\begin{align}
\frac{\|\matA\|_2}{\sqrt{\rank(\matA)}}
\leq
\onorm{\matA}_2
&\leq \|\matA\|_2.
\end{align}
\end{lemma}
\begin{proof}
The proof is trivial for $\matA=\matzero$. If $\matA\neq\matzero$, set $r=\rank(\matA)$ and let  $\sigma_1,\dots,\sigma_r$ denote the nonzero singular values of $\matA$ organized in decreasing order.  
Unitary invariance of $\onorm{\cdot}_2$ and $\|\cdot\|_2$ (cf. \cite[Problem 5, p. 311]{hojo13}) yields 
$\onorm{\matA}_2=\sigma_1$ and $\|\matA\|_2=\sqrt{\sum_{i=1}^r\sigma_i^2}$.  
The claim now follows from   
\begin{align}
\sigma_1\leq \sqrt{\sum_{i=1}^r\sigma_i^2}\leq \sqrt{r}\sigma_1. 
\end{align}  
\end{proof}

\begin{lemma}\label{lem:mmm}
For $\matA=(\veca_1\dots\veca_n)\in\complex^{m\times n}$, we have 
\begin{align}
\onorm{\matA}_1=\max_{j\in\{1,\dots,n\}}\|\veca_j\|_1\label{eq:opnorm1a}
\end{align}
and 
\begin{align}
\frac{1}{n}\|\matA\|_1\leq \onorm{\matA}_1\leq \|\matA\|_1. \label{eq:opnorm1b}
\end{align}
\end{lemma}
\begin{proof}
The identity \eqref{eq:opnorm1a} is established in \cite[p.294]{hojo13}, and \eqref{eq:opnorm1b} follows directly from \eqref{eq:opnorm1a}. 
\end{proof}







  \bibliographystyle{IEEEtran} 

  \bibliography{references}\label{refs}

  \cleardoublepage



%

\end{document}